\documentclass[a4paper,UKenglish,final]{lipics}

\usepackage{amsfonts,amssymb,amsmath}
\usepackage{xspace}
\usepackage{xparse}
\usepackage{relsize}
\usepackage{tikz}
\usetikzlibrary{arrows,automata,shapes,calc,positioning,intersections}
\usepackage{paralist}
\usepackage[strict]{changepage}
\usepackage[ruled,boxed,vlined,linesnumbered]{algorithm2e}

\usepackage[obeyFinal]{todonotes}

\newcommand\N{\ensuremath{\mathbf{N}}}
\newcommand\R{\ensuremath{\mathbf{R}}}
\newcommand\Rplus{\ensuremath{\R^+}}
\newcommand\Rbar{\ensuremath{\overline{\R}}}
\newcommand\Z{\ensuremath{\mathbf{Z}}}
\newcommand\Q{\ensuremath{\mathbf{Q}}}
\renewcommand\vec[1]{\boldsymbol{#1}}

\renewcommand\leq{\leqslant}
\renewcommand\geq{\geqslant}

\newcommand\MinPl{\ensuremath{\mathsf{Min}}\xspace} 
\newcommand\MaxPl{\ensuremath{\mathsf{Max}}\xspace} 

\newcommand\Guard[1]{\ensuremath{\mathsf{Guard}(#1)}} 
\newcommand\regcst[1]{\ensuremath{\mathsf{Reg}_{#1}}} 
\newcommand\den[1]{\ensuremath{[\![#1]\!]\xspace}} 

\newcommand\valuation{\ensuremath{\nu}\xspace}
\newcommand\CF[1]{\ensuremath{{\sf CF}_{#1}}} 
\newcommand\opcf{\ensuremath{\rhd}} 

\newcommand\PTG{{\sf PTG}\xspace}

\newcommand\Locs{\ensuremath{L}}
\newcommand\locs{\Locs}
\newcommand\LocsMin{\Locs_{\MinPl}}
\newcommand\LocsMax{\Locs_{\MaxPl}}
\newcommand\LocsFin{\ensuremath{\Locs_f}} 
\newcommand\LocsUrg{\ensuremath{\Locs_u}} 
\newcommand\loc{\ensuremath{\ell}}

\newcommand\fgoal{\ensuremath{\varphi}}
\newcommand\fgoalvec{\ensuremath{\boldsymbol{\varphi}}}
\newcommand\transitions{\ensuremath{\Delta}}
 
\newcommand\price{\ensuremath{\pi}}
\newcommand\transition{\delta}
\newcommand\game{\ensuremath{\mathcal G}\xspace}
\newcommand\hame{\ensuremath{\mathcal H}\xspace}
\newcommand\cstgame{\ensuremath{S_\game}}
\newcommand\reggame{\regcst{\game}}

\newcommand\maxPriceLoc{\ensuremath{\Pi^{\mathrm{loc}}}\xspace}
\newcommand\maxPriceTrans{\ensuremath{\Pi^{\mathrm{tr}}}\xspace}
\newcommand\maxPriceFin{\ensuremath{\Pi^{\mathrm{fin}}}\xspace}

\newcommand\confs[1]{\ensuremath{{\sf Conf}_{#1}}\xspace}
\newcommand\confgame{\confs{\game}}

\newcommand\costname{{\sf Cost}}
\newcommand\costgame[2]{\ensuremath{\costname_{#1}(#2)}}
\newcommand\cost[1]{\ensuremath{\costname(#1)}}

\newcommand\len[1]{\ensuremath{|#1|}}
\newcommand\run{\rho}

\newcommand\strat{\ensuremath{\sigma}\xspace}
\newcommand\stratmin{\strat_{\MinPl}}
\newcommand\stratmax{\strat_{\MaxPl}}
\newcommand\stratsofmin{\ensuremath{{\sf Strat}_{\MinPl}}}
\newcommand\stratsofmax{\ensuremath{{\sf Strat}_{\MaxPl}}}
\newcommand\stratsofmingame[1]{\stratsofmin(#1)}
\newcommand\stratsofmaxgame[1]{\stratsofmax(#1)}
\newcommand\outcomes{\mathsf{Play}}
\newcommand\Play[1]{\ensuremath{\outcomes(#1)}}

\newcommand\points{\mathsf{pts}}
\newcommand\intervals{\mathsf{int}}
\newcommand\fakeValue{\mathsf{fake}}

\newcommand\val{\ensuremath{{\sf Val}}}
\newcommand\Val{\val}
\newcommand\uval{\ensuremath{\overline{\val}}}
\newcommand\lval{\ensuremath{\underline{\val}}}

\newcommand\Value{\val}
\newcommand\uvalgs[2]{\uval_{#1}^{#2}}
\newcommand\lvalgs[2]{\lval_{#1}^{#2}}
\newcommand\valgs[2]{\val_{#1}^{#2}}
\newcommand\uvalgame{\uvalgs{\game}{}}
\newcommand\lvalgame{\lvalgs{\game}{}}
\newcommand\valgame{\valgs{\game}{}}

\newcommand\SPTG{{\sf SPTG}\xspace}
\newcommand\rightpoint{\ensuremath{r}\xspace}

\SetKwFunction{Waiting}{wait}
\SetKwFunction{SolveInstant}{solveInstant}
\SetKwFunction{Solve}{solve}
\newcommand\locMin{\ensuremath{\loc^\star}}
\newcommand\Next{\textsf{left}}

\newcommand\operator{\mathcal F}

\newcommand\bupval[1]{\overline{\Value}^{\leq #1}}
\newcommand\minstrategy{\stratmin}
\newcommand\maxstrategy{\stratmax}
\newcommand\Zstar{\ensuremath{\Z_{\valuation,\fgoalvec}}}
\newcommand\Zstarinf{\ensuremath{\Z_{\valuation,\fgoalvec}^{+\infty}}}
\newcommand\costbound[1]{\costname^{\leq#1}}
\newcommand\possval{\ensuremath{{\sf PossVal}}}
\newcommand\posscp{{\sf PossCP}}
\newcommand{\F}{{\sf F}}

\title{Simple Priced Timed Games Are Not That Simple\footnote{The
    research leading to these results has received funding from the
    European Union Seventh Framework Programme (FP7/2007-2013) under
    Grant Agreement n°601148 (CASSTING).}}

\author[1]{Thomas Brihaye} \author[2]{Gilles Geeraerts}
\author[1]{Axel Haddad} \author[3]{Engel Lefaucheux}
\author[4]{Benjamin Monmege} 

\affil[1]{Universit\'e de Mons, Belgium,
  \texttt{\{thomas.brihaye,axel.haddad\}@umons.ac.be}}
\affil[2]{Universit\'e libre de Bruxelles, Belgium,
  \texttt{gigeerae@ulb.ac.be}}
\affil[3]{LSV, ENS Cachan, Inria Rennes, France,
  \texttt{engel.lefaucheux@ens-cachan.fr}}
\affil[4]{LIF, Aix-Marseille Universit\'e, CNRS, France,
\texttt{benjamin.monmege@lif.univ-mrs.fr}}

\authorrunning{T. Brihaye, G. Geeraerts, A. Haddad, E. Lefaucheux,
  B. Monmege}

\Copyright{Thomas Brihaye, Gilles Geeraerts, Axel Haddad, Engel Lefaucheux,
  Benjamin Monmege}

\theoremstyle{theorem}
\newtheorem{proposition}[theorem]{Proposition}

\begin{document}

\maketitle

\begin{abstract}
  Priced timed games are two-player zero-sum games played on priced
  timed automata (whose locations and transitions are labeled by
  weights modeling the costs of spending time in a state and executing
  an action, respectively). The goals of the players are to minimise
  and maximise the cost to reach a target location, respectively. We
  consider priced timed games with one clock and arbitrary (positive
  and negative) weights and show that, for an important subclass of
  theirs (the so-called \emph{simple} priced timed games), one can
  compute, in exponential time, the optimal values that the players
  can achieve, with their associated optimal strategies. As side
  results, we also show that one-clock priced timed games are
  determined and that we can use our result on simple priced timed
  games to solve the more general class of so-called reset-acyclic
  priced timed games (with arbitrary weights and one-clock).
\end{abstract}

\section{Introduction} 

The importance of models inspired from the field of game theory is
nowadays well-established in theoretical computer science. They allow
to describe and analyse the possible interactions of antagonistic
agents (or players) as in the \emph{controller synthesis} problem, for
instance. This problem asks, given a model of the environment of a
system, and of the possible actions of a controller, to compute a
controller that constraints the environment to respect a given
specification. Clearly, one can not, in general, assume that the two
players (the environment and the controller) will collaborate, hence
the need to find a \emph{controller strategy} that enforces the
specification \emph{whatever the environment does}. This question thus
reduces to computing a so-called winning strategy for the
corresponding player in the game model.

In order to describe precisely the features of complex computer
systems, several game models have been considered in the
literature. In this work, we focus on the model of Priced Timed
Games~\cite{RW89} (\PTG{s} for short), which can be regarded as an
extension (in several directions) of classical finite automata. First,
like timed automata~\cite{AluDil94}, \PTG{s} have \emph{clocks}, which
are real-valued variables whose values evolve with time elapsing, and
which can be tested and reset along the transitions. Second, the
locations are associated with price-rates and transitions are labeled
by discrete prices, as in priced timed
automata~\cite{BehFeh01,AluLa-04,BouBri07}. These prices allow one to
associate a \emph{cost} with all runs (or plays), which depends on the
sequence of transitions traversed by the run, and on the time spent in
each visited location. Finally, a \PTG is played by two players,
called $\MinPl$ and $\MaxPl$, and each location of the game is owned
by either of them (we consider a turn-based version of the game). The
player who controls the current location decides how long to wait, and
which transition to take.

In this setting, the goal of $\MinPl$ is to reach a given set of
target locations, following a play whose cost is as small as
possible. Player $\MaxPl$ has an antagonistic objective: he tries to
avoid the target locations, and, if not possible, to maximise the
accumulated cost up to the first visit of a target location. To
reflect these objectives, we define the upper value $\uval$ of the
game as a mapping of the configurations of the \PTG to the least cost
that $\MinPl$ can guarantee while reaching the target, whatever the
choices of $\MaxPl$. Similarly, the lower value $\lval$ returns the
greatest cost that $\MaxPl$ can ensure (letting the cost being
$+\infty$ in case the target locations are not reached).

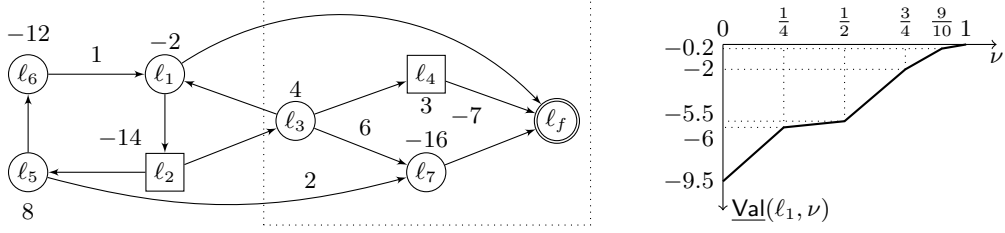
\begin{figure}[tbp]
  \begin{center}
    \begin{tabular}{ccc}
      
      \begin{tikzpicture}[minimum size=5mm,node distance=1.3cm]
        \everymath{\footnotesize}
        
        \node[draw,circle,label={[label distance=-1mm]above:$-2$}] (q1) {\makebox[0pt][c]{$\loc_1$}};
        \node[draw,rectangle,below of=q1,label={[label distance=-1mm]above left:$-14$}] (q2) {\makebox[0pt][c]{$\loc_2$}};
        \node[draw,circle,below right of=q1,xshift=8mm,yshift=3mm,label={[label distance=-1mm]above:$4$}] (q3) {\makebox[0pt][c]{$\loc_3$}};
        \node[draw,rectangle,above right of=q3,xshift=8mm,yshift=-3mm,label={[label distance=-1mm]below:$3$}] (q4) {\makebox[0pt][c]{$\loc_4$}};
        \node[draw,circle,left of=q2,label=below:$8$,xshift=-5mm] (q5) {\makebox[0pt][c]{$\loc_5$}};
        \node[draw,circle,above of=q5,label=above:$-12$] (q6) {\makebox[0pt][c]{$\loc_6$}};
        \node[draw,circle,below of=q4,label={[label distance=-1mm]above:$-16$}] (q7){\makebox[0pt][c]{$\loc_7$}};
        \node[draw,circle,below right of=q4,xshift=8mm,yshift=3mm,accepting] (qf) {\makebox[0pt][c]{$\loc_f$}};
        
        \path[draw,arrows=-latex'] (q1) edge (q2) 
        (q2) edge (q3) 
        (q2) edge (q5)
        (q5) edge (q6) 
        (q6) edge node[above] {$1$} (q1)
        (q5) edge[bend right=15] node[pos=.7,above,xshift=2mm] {$2$} (q7) 
        (q3) edge node[above right,xshift=-2mm] {$6$} (q7)
        (q3) edge (q1)
        (q3) edge (q4) 
        (q4) edge node[below left] {$-7$} (qf)
        (q7) edge (qf) 
        (q1) edge[bend left=40] (qf);
        
        \draw[dotted] (1.3,-2) rectangle (5.6,1) ; 
      \end{tikzpicture}
      &&
         \begin{tikzpicture}[xscale=.8,yscale=0.55]
           \draw[->] (6,-5) -- (10.5,-5) node[anchor=north] {$\valuation$};
           \draw	(6,-5) node[anchor=south] {$0$}
           (7,-5) node[anchor=south] {$\frac 1 4$}
           (8,-5) node[anchor=south] {$\frac 1 2$}
           (9,-5) node[anchor=south] {$\frac 3 4$}
           (9.6,-5) node[anchor=south] {$\frac 9 {10}$}
           (10,-5) node[anchor=south] {$1$};
           
           \draw[->] (6,-5) -- (6,-9) node[anchor=west] {$\lval(\loc_1,\valuation)$};
           \draw	(6,-8.3) node[anchor=east] {$-9.5$}
           (6,-7.3) node[anchor=east] {$-6$}
           (6,-6.7) node[anchor=east] {$-5.5$}
           (6,-5.6) node[anchor=east] {$-2$}
           (6,-5.1) node[anchor=east] {$-0.2$};
           
           \draw[thick] (6,-8.3) -- (7,-7) --
           (8,-6.85)--(9,-5.6)--(9.6,-5.1)--(10,-5);

           \draw[dotted] (7,-7) -- (6,-7);
           \draw[dotted] (7,-7) -- (7,-5);
           \draw[dotted] (8,-6.85) -- (6,-6.85) ;
           \draw[dotted] (8,-6.85) -- (8,-5) ;
           \draw[dotted] (9,-5.6) -- (6,-5.6) ;
           \draw[dotted] (9,-5.6) -- (9, -5) ;
           \draw[dotted] (9.6,-5.1) -- (6,-5.1) ;
           \draw[dotted] (9.6,-5.1) -- (9.6,-5) ;

         \end{tikzpicture}
    \end{tabular}
    \caption{A simple priced timed game (left) and the lower value
      function of location $\loc_1$ (right).}
    \label{fig:ex-ptg2}
  \end{center}
\end{figure}
An example of \PTG is given in Figure~\ref{fig:ex-ptg2}, where the
locations of $\MinPl$ (respectively, $\MaxPl$) are represented by
circles (respectively, rectangles), and the integers next to the
locations are their price-rates, i.e., the cost of spending one time
unit in the location. Moreover, there is only one clock $x$ in the
game, which is never reset and all guards on transitions are
$x\in[0,1]$ (hence this guard is not displayed and transitions are
only labeled by their respective discrete cost): this is an example of
\emph{simple priced timed game}, as we will define them properly
later. It is easy to check that $\MinPl$ can force reaching the target
location $\ell_f$ from all configurations $(\ell,\valuation)$ of the
game, where $\ell$ is a location and $\valuation$ is a real valuation
of the clock in $[0,1]$. Let us comment on the optimal strategies for
both players.  From a configuration $(\loc_4,\valuation)$, with
$\valuation\in[0,1]$, $\MaxPl$ better waits until the clock takes
value $1$, before taking the transition to $\loc_f$ (he is forced to
move, by the rule of the game). Hence, $\MaxPl$'s optimal value is
$3(1-\valuation)-7 = -3\valuation -4$ from all configurations
$(\loc_4,\valuation)$. Symmetrically, it is easy to check that
$\MinPl$ better waits as long as possible in $\loc_7$, hence his
optimal value is $-16(1-\valuation)$ from all configurations
$(\loc_7,\valuation)$. However, optimal value functions are not always
\emph{that simple}, see for instance the lower value function of
$\loc_1$ on the right of Figure~\ref{fig:ex-ptg2}, which is a
piecewise affine function. To understand why value functions can be
piecewise affine, consider the sub-game enclosed in the dotted
rectangle in Figure~\ref{fig:ex-ptg2}, and consider the value that
$\MinPl$ can guarantee from a configuration of the form
$(\loc_3,\valuation)$ in this sub-game. Clearly, $\MinPl$ must decide
how long he will spend in $\loc_3$ and whether he will go to $\loc_4$
or $\loc_7$. His optimal value from all $(\loc_3,\valuation)$ is thus
$\inf_{0\leq t\leq 1-\valuation} \min\big(4t+ (-3(\valuation+t)-4), 4t
+ 6 -16(1-(\valuation+t))\big) = \min
(-3\valuation-4,16\valuation-10)$.
Since $16\valuation-10\geq -3\valuation-4$ if and only if
$\valuation \leq 6/19$, the best choice of $\MinPl$ is to move
instantaneously to $\loc_7$ if $\valuation \in [0,6/19]$ and to move
instantaneously to $\loc_4$ if $\valuation \in (6/19,1]$, hence the
value function of $\loc_3$ (in the subgame) is a piecewise affine
function with two pieces.

\subparagraph*{Related work.} \PTG{s} were independently investigated
in \cite{BouCas04} and \cite{AluBer04}. For (non-necessarily
turn-based) \PTG{s} with \emph{non-negative} prices, semi-algorithms
are given to decide the \emph{value problem} that is to say, whether
the lower value of a location (the best cost that $\MinPl$ can
guarantee in valuation $0$), is below a given threshold. They also
showed that, under the \emph{strongly non-Zeno assumption} on prices
(asking the existence of $\kappa>0$ such that every cycle in the
underlying region graph has a cost at least $\kappa$), the proposed
semi-algorithms always terminate. This assumption was justified in
\cite{BriBru05,BouBri06} by showing that, in the absence of non-Zeno
assumption, the \emph{existence problem}, that is to decide whether
$\MinPl$ has a strategy guaranteeing to reach a target location with a
cost below a given threshold, is indeed undecidable for \PTG{s} with
non-negative prices and three or more clocks. This result was recently
extended in \cite{BJMr14} to show that the \emph{value problem} is
also undecidable for \PTG{s} with non-negative prices and four or more
clocks. In \cite{BCJ09}, the undecidability of the existence problem
has also been shown for \PTG{s} with arbitrary price-rates (without
prices on transitions), and two or more clocks.
On a positive side, the value problem was shown decidable by
\cite{BouLar06} for \PTG{s} with one clock when the prices are
non-negative: a 3-exponential time algorithm was first proposed,
further refined in \cite{Rut11,DueIbs13} into an exponential time
algorithm. The key point of those algorithms is to reduce the problem
to the computation of optimal values in a restricted family of \PTG{s}
called \emph{Simple Priced Timed Games} (\SPTG{s} for short), where
the underlying automata contain no guard, no reset, and the play is
forced to stop after one time unit.  More precisely, the \PTG is
decomposed into a sequence of \SPTG{s} whose value functions are
computed and re-assembled to yield the value function of the original
\PTG.  Alternatively, and with radically different techniques, a
pseudo-polynomial time algorithm to solve one-clock \PTG{s} with
arbitrary prices on transitions, and price-rates restricted to two
values amongst $\{-d,0,+d\}$ (with $d\in\N$) was given in
\cite{BGKMMT14}.

\subparagraph*{Contributions.} Following the decidability results
sketched above, we consider \PTG{s} with one clock. We extend those
results by considering arbitrary (positive and negative)
prices. Indeed, all previous works on \PTG{s} with only one clock
(except \cite{BGKMMT14}) have considered non-negative weights only,
and the status of the more general case with arbitrary weights has so
far remained elusive. Yet, arbitrary weights are an important modeling
feature. Consider, for instance, a system which can consume but also
produce energy at different rates. In this case, energy consumption
could be modeled as a positive price-rate, and production by a
negative price-rate. We propose an \emph{exponential time algorithm to
  compute the value of one-clock \SPTG{s} with arbitrary
  weights}. While this result might sound limited due to the
restricted class of simple \PTG{s} we can handle, we recall that the
previous works mentioned above \cite{BouLar06,Rut11,DueIbs13} have
demonstrated that solving \SPTG{s} is a key result towards solving
more general \PTG{s}. Moreover, this algorithm is, as far as we know,
the first to handle the full class of \SPTG{s} with arbitrary weights,
and we note that the solutions (either the algorithms or the proofs)
known so far do not generalise to this case. Finally, as a side
result, this algorithm allows us to solve the more general class of
\emph{reset-acyclic} one-clock \PTG{s} that we introduce. Thus,
although we can not (yet) solve the whole class of \PTG{s} with
arbitrary weights, our result may be seen as a potentially important
milestone towards this goal.

Some proofs and technical details are in the Appendix.

\section{Priced timed games: syntax, semantics, and preliminary
  results}
\label{sec:ptg}

\subparagraph*{Notations and definitions.}  Let $x$ denote a positive
real-valued variable called \emph{clock}. A \emph{guard} (or
\emph{clock constraint}) is an interval with endpoints in
$\N\cup\{+\infty\}$. We often abbreviate guards, for instance
$x\leq 5$ instead of $[0,5]$.  Let $S\subseteq \Guard{x}$ be a finite
set of guards. We let $\den{S}=\bigcup_{I\in S} I$. Assuming
$M_0=0<M_1<\cdots<M_k$ are all the endpoints of the intervals in~$S$
(to which we add $0$), we let
$\regcst{S}=\{(M_i,M_{i+1})\mid 0\leq i\leq k-1\} \cup \{ \{M_i\}\mid
0\leq i\leq k\}$
be the set of \emph{regions} of $S$. Observe that $\regcst{S}$ is also
a set of guards.

We rely on the notion of \emph{cost
  function} to formalise the notion of optimal value function sketched
in the introduction.
Formally, for a set of guards $S\subseteq \Guard{x}$, a \emph{cost
  function} over $S$ is a function
$f\colon \den{\regcst{S}} \to \Rbar= \R\cup \{+\infty,-\infty\}$ such
that over all regions $r\in\regcst{S}$, $f$ is either infinite or a
continuous piecewise affine function, with a finite set of cutpoints
(points where the first derivative is not defined)
$\{\kappa_1,\ldots,\kappa_p\}\subseteq \Q$, and with
$f(\kappa_i)\in\Q$ for all $1\leq i\leq p$. In particular, if
$f(r)=\{f(\valuation)\mid \valuation\in r\}$ contains $+\infty$
(respectively, $-\infty$) for some region $r$, then $f(r)=\{+\infty\}$
($f(r)=\{-\infty\}$). We denote by $\CF{S}$ the set of all cost
functions over $S$. In our algorithm to solve \SPTG{s}, we will need
to combine cost functions thanks to the $\opcf$ operator. Let
$f\in\CF{S}$ and $f'\in\CF{S'}$ be two costs functions on set of
guards $S,S'\subseteq \Guard{x}$, such that $\den S \cap \den{S'}$ is
a singleton. We let $f\opcf f'$ be the cost function in
$\CF{S\cup S'}$ such that $(f\opcf f')(\valuation)=f(\valuation)$ for
all $\valuation\in\den{\regcst S}$, and
$(f\opcf f')(\valuation)=f'(\valuation)$ for all
$\valuation\in\den{\regcst{S'}}\setminus\den{\regcst{S}}$.

We consider an extended notion of one-clock priced timed games
(\PTG{s} for short) allowing for the use of \emph{urgent locations},
where only a zero delay can be spent, and \emph{final cost functions}
which are associated with each final location and incur an extra cost
to be paid when ending the game in this location. Formally, a \PTG
$\game$ is a tuple $(\LocsMin, \LocsMax, \LocsFin, \LocsUrg, \fgoalvec,
\transitions, \price)$ where
\begin{inparaenum}[]
\item $\LocsMin$ (respectively, $\LocsMax$) is a finite set of
  \emph{locations} for player $\MinPl$ (respectively, $\MaxPl$), with
  $\LocsMin\cap\LocsMax=\emptyset$;
\item $\LocsFin$ is a finite set of \emph{final} locations, and we let
  $\Locs = \LocsMin \cup \LocsMax\cup\LocsFin$ be the whole location
  space;
\item $\LocsUrg \subseteq \Locs\setminus\LocsFin$ indicates
  \emph{urgent} locations\footnote{Here we differ from~\cite{BouLar06}
    where $\LocsUrg \subseteq \LocsMax$.}; 
\item
  $\transitions \subseteq (\Locs\setminus \LocsFin)\times \Guard{x}
  \times \{\top,\bot\} \times \Locs$
  is a finite set of \emph{transitions};
\item $\fgoalvec = (\fgoal_\loc)_{\loc\in\LocsFin}$ associates to each
  $\loc\in \LocsFin$ its \emph{final cost function}, that is an
  affine\footnote{The affine restriction on final cost function is to
    simplify our further arguments, though we do believe that all of
    our results could be adapted to cope with general cost functions.}
  cost function $\fgoal_\loc$ over
  $\cstgame = \{I\mid \exists \loc, R, \loc': (\loc,I,R,\loc')\in
  \transitions\}$;
\item $\price\colon \Locs \cup \transitions \to \Z$ mapping an integer
  \emph{price} to each location---its \emph{price-rate}---and
  transition.
\end{inparaenum}

Intuitively, a transition $(\loc,I,R,\loc')$ changes the current
location from $\loc$ to $\loc'$ if the clock has value in $I$ and the
clock is reset according to the Boolean $R$. We assume that, in all
\PTG{s}, the clock $x$ is \emph{bounded}, i.e., there is $M\in\N$ such
that for all guards $I\in \cstgame$,
$I\subseteq [0,M]$.\footnote{Observe that this last restriction is
  \emph{not} without loss of generality in the case of \PTG{s}. While
  all timed automata $\mathcal{A}$ can be turned into an equivalent
  (with respect to reachability properties) $\mathcal{A}'$ whose
  clocks are bounded \cite{BehFeh01}, this technique can not be
  applied to \PTG{s}, in particular with arbitrary prices.} We denote
by $\reggame$ the set $\regcst{\cstgame}$ of \emph{regions of}
$\game$.  We further denote\footnote{Throughout the paper, we often
  drop the $\game$ in the subscript of several notations when the game
  is clear from the context.} by $\maxPriceTrans_\game$,
$\maxPriceLoc_\game$ and $\maxPriceFin_\game$ respectively the values
$\max_{\transition\in\transitions}|\price(\transition)|$,
$\max_{\loc\in\Locs} |\price(\loc)|$ and
$\sup_{\valuation\in [0,M]} \max_{\loc\in\Locs}
|\fgoal_\loc(\valuation)|=\max_{\loc\in\Locs} \max(
|\fgoal_\loc(0)|,|\fgoal_\loc(M)|)$.
That is, $\maxPriceTrans_\game$, $\maxPriceLoc_\game$ and
$\maxPriceFin_\game$ are the largest absolute values of the location
prices, transition prices and final cost functions.

Let
$\game = (\LocsMin, \LocsMax, \LocsFin, \LocsUrg, \fgoalvec,
\transitions, \price)$
be a \PTG. A \emph{configuration} of $\game$ is a pair
$s=(\loc,\valuation)\in \Locs \times \Rplus$. We denote by \confgame
the set of configurations of $\game$. Let $(\loc,\valuation)$ and
$(\loc',\valuation')$ be two configurations. Let
$\transition=(\loc,I,R,\loc')\in\transitions$ be a transition of
$\game$ and $t\in\Rplus$ be a delay. Then, there is a
$(t,\transition)$-transition from $(\loc,\valuation)$ to
$(\loc',\valuation')$ with cost $c$, denoted by
$(\loc, \valuation) \xrightarrow{t,\transition,c}(\loc',
\valuation')$, if \begin{inparaenum}[$(i)$]
\item $\loc\in\LocsUrg$ implies $t=0$; 
\item $\valuation+t\in I$; 
\item $R=\top$ implies $\valuation'=0$; 
\item $R=\bot$ implies $\valuation'=\valuation+t$;
\item $c=\price(\transition)+t\times\price(\loc)$.
\end{inparaenum}
Observe that the cost of $(t,\transition)$ takes into account the
price-rate of~$\loc$, the delay spent in $\loc$, and the price of
$\transition$. We assume that the game has no deadlock: for all
$s\in\confgame$, there are $(t,\transition,c)$ and $s'\in\confgame$
such that $s\xrightarrow{t,\transition,c} s'$. Finally, we write
$s\xrightarrow{c} s'$ whenever there are $t$ and $\transition$ such
that $s\xrightarrow{t,\transition,c} s'$.
A \emph{play} of \game is a finite or infinite path
$\rho = (\loc_0,\valuation_0)\xrightarrow{c_0}
(\loc_1,\valuation_1)\xrightarrow{c_1}(\loc_2,\valuation_2) \cdots$.
For a finite play
$\rho = (\loc_0,\valuation_0)\xrightarrow{c_0}
(\loc_1,\valuation_1)\xrightarrow{c_1}(\loc_2,\valuation_2) \cdots
\xrightarrow{c_{n-1}} (\loc_n,\valuation_n)$,
we let $\len{\rho}=n$. For an infinite play
$\rho = (\loc_0,\valuation_0)\xrightarrow{c_0}
(\loc_1,\valuation_1)\xrightarrow{c_1}(\loc_2,\valuation_2) \cdots$,
we let $\len{\rho}$ be the least position $i$ such that
$\loc_i\in \LocsFin$ if such a position exists, and
$\len{\rho}=+\infty$ otherwise. Then, we let $\costgame{\game}{\rho}$
be the \emph{cost} of $\rho$, with $\costgame{\game}{\rho}=+\infty$ if
$\len{\rho}=+\infty$, and
$\costgame{\game}{\rho}=\sum^{{|\rho|}-1}_{i=0} c_i +
\fgoal_{\loc_{|\rho|}}(\valuation_{|\rho|})$
otherwise. 

A \emph{strategy} for player $\MinPl$ is a function $\stratmin$
mapping every finite play ending in location of $\MinPl$ to a pair
$(t,\transition)\in \Rplus\times\transitions$, indicating what
$\MinPl$ should play. We also request that the strategy proposes only
valid pairs $(t,\transition)$, i.e., that for all runs $\run$ ending
in $(\loc,\valuation)$, $\stratmin(\run) = (t,(\loc,I,R,\loc'))$
implies that $\valuation + t \in I$. Strategies $\stratmax$ of player
$\MaxPl$ are defined accordingly. We let $\stratsofmingame{\game}$ and
$\stratsofmaxgame{\game}$ be the sets of strategies of $\MinPl$ and
$\MaxPl$, respectively. A pair of strategies
$(\stratmin,\stratmax)\in \stratsofmingame{\game}\times
\stratsofmaxgame{\game}$
is called a \emph{profile of strategies}. Together with an initial
configuration $s_0=(\loc_0,\valuation_0)$, it defines a unique play
$\Play{s_0,\stratmin,\stratmax}=
s_0\xrightarrow{c_0}s_1\xrightarrow{c_1} s_2 \cdots s_k
\xrightarrow{c_k} \cdots$
where for all $j\geq 0$, $s_{j+1}$ is the unique configuration such
that $s_j\xrightarrow{t_j,\transition_j,c_j}s_{j+1}$ with
$(t_j,\transition_j)=\stratmin(s_0\xrightarrow{c_0}s_1\cdots
s_{j-1}\xrightarrow{c_{j-1}} s_{j})$
if $\loc_j\in\LocsMin$; and
$(t_j,\transition_j)=\stratmax(s_0\xrightarrow{c_0}s_1\cdots
s_{j-1}\xrightarrow{c_{j-1}} s_{j})$
if $\loc_j\in\LocsMax$.  We let $\Play{\stratmin}$ (respectively,
$\Play{s_0,\stratmin}$) be the set of plays that conform with
$\stratmin$ (and start in~$s_0$).

As sketched in the introduction, we consider optimal reachability-cost
games on \PTG{s}, where the aim of player $\MinPl$ is to reach a
location of $\LocsFin$ while minimising the cost. To formalise this
objective, we let the value of a strategy $\stratmin$ for $\MinPl$ be
the function $\valgs{\game}{\stratmin}\colon \confgame\to\Rbar$ such
that for all $s\in\confgame$:
$\valgs{\game}{\stratmin}(s)= \sup_{\stratmax\in\stratsofmax}
\cost{\Play{s,\stratmin,\stratmax}}$.
Intuitively, $\valgs{\game}{\stratmin}(s)$ is the largest value that
$\MaxPl$ can achieve when playing against strategy $\stratmin$ of
$\MinPl$ (it is thus a worst case from the point of view of
$\MinPl$). Symmetrically, for $\stratmax\in\stratsofmax$,
$\valgs{\game}{\stratmax}(s)= \inf_{\stratmin\in\stratsofmin}
\cost{\Play{s,\stratmin,\stratmax}}$,
for all $s\in\confgame$. Then, the \emph{upper and lower values} of
$\game$ are respectively the functions
$\uvalgame\colon \confgame\to\Rbar$ and
$\lvalgame\colon \confgame\to\Rbar$ where, for all $s\in \confgame$,
$\uvalgame(s)= \inf_{\stratmin\in\stratsofmin}
\valgs{\game}{\stratmin}(s)$
and
$\lvalgame(s)= \sup_{\stratmax\in\stratsofmax}
\valgs{\game}{\stratmax}(s)$.
We say that a game is \emph{determined} if the lower and the upper
values match for every configuration $s$, and in this case, we say
that the optimal value $\valgame$ of the game $\game$ exists, defined
by $\valgame = \lvalgame = \uvalgame$.  A strategy $\stratmin$ of
$\MinPl$ is \emph{optimal} (respectively,
$\varepsilon$-\emph{optimal}) if $\valgs{\game}{\stratmin}=\uvalgame$
($\valgs{\game}{\stratmin}\leq \uvalgame+\varepsilon$), i.e.,
$\stratmin$ ensures that the cost of the plays will be at most
$\uvalgame$ ($\uvalgame+\varepsilon$). Symmetrically, a strategy
$\stratmax$ of $\MaxPl$ is \emph{optimal} (respectively,
$\varepsilon$-\emph{optimal}) if $\valgs{\game}{\stratmax}=\lvalgame$
($\valgs{\game}{\stratmax}\geq\lvalgame-\varepsilon$).

\subparagraph*{Properties of the value.} Let us now prove useful
preliminary properties of the value function of \PTG{s}, that---as far
as we know---had hitherto never been established. Using a general
determinacy result by Gale and Stewart \cite{GS1953}, we can show that
\PTG{s} (with one clock) are \emph{determined}. Hence, the value
function $\valgame$ exists for all \PTG $\game$. We can further show
that, for all locations~$\loc$, $\valgame(\loc)$ is a \emph{piecewise
  continuous function} that might exhibit discontinuities \emph{only
  on the borders of the regions} of $\reggame$ (where $\valgame(\loc)$
is the function such that
$\valgame(\loc)(\valuation)=\valgame(\loc,\valuation)$ for all
$\valuation\in\Rplus$). See Appendix~\ref{app:continuity-of-val} for
detailed proofs of these results. The continuity holds only in the
case of \PTG{s} with a single clock. An example with two clocks and a
value function exhibiting discontinuities inside a region is in
Appendix~\ref{app:pas-continu2}.

\begin{theorem}
  \label{thm:determined} \label{prop:continuity-of-val} For all
  (one-clock) \PTG{s} $\game$:
  \begin{inparaenum}[$(i)$]
  \item $\uvalgame=\lvalgame$, i.e., \PTG{s} are \emph{determined}; and
  \item for all $r\in \reggame$, for all $\loc\in\locs$,
    $\valgame(\loc)$ is either infinite or continuous over $r$.
  \end{inparaenum}
\end{theorem}

\subparagraph*{Simple priced timed games.} As sketched in the
introduction, our main contribution is to solve the special case of
simple one-clock priced timed games with arbitrary costs. Formally, an
$\rightpoint$-\SPTG, with $\rightpoint\in \Q^+\cap [0,1]$, is a \PTG
$\game = (\LocsMin, \LocsMax, \LocsFin, \LocsUrg, \fgoalvec,
\transitions, \price)$
such that for all transitions $(\loc,I,R,\loc')\in \transitions$, 
$I=[0,\rightpoint]$ and $R=\bot$. Hence, transitions of
$\rightpoint$-\SPTG{s} are henceforth denoted by $(\loc,\loc')$,
dropping the guard and the reset. Then, an \SPTG is a $1$-\SPTG. This
paper is devoted mainly to proving the following theorem on \SPTG{s}:

\begin{theorem}\label{thm:main-result}
  Let $\game$ be an \SPTG. Then, for all locations $\loc \in \locs$,
  the function $\valgame(\loc)$ is either infinite, or continuous and
  piecewise-affine with at most an exponential number of
  cutpoints. The value functions for all locations, as well as a pair
  of optimal strategies $(\stratmin,\stratmax)$ (that always exist if
  no values are infinite) can be computed in exponential time.
\end{theorem}

Before sketching the proof of this theorem, we discuss a class of
(simple) strategies that are sufficient to play optimally. Roughly
speaking, \MaxPl has always a \emph{memoryless} optimal strategy,
while \MinPl might need \emph{(finite) memory} to play optimally---it
is already the case in untimed quantitative reachability games with
arbitrary weights (see Appendix~\ref{app:ex-JF}). Moreover, these
strategies are finitely representable (recall that even a memoryless
strategy depends on the current \emph{configuration} and that there
are infinitely many in our time setting).

We formalise $\MaxPl$'s strategies with the notion of \emph{finite
  positional strategy} (FP-strategy): they are memoryless strategies
$\strat$ (i.e., for all finite plays
$\rho_1 = \rho_1' \xrightarrow{c_1} s$ and
$\rho_2 = \rho_2' \xrightarrow{c_2} s$ ending in the same
configuration, we have $\strat(\rho_1) = \strat(\rho_2)$), such that
for all locations $\loc$, there exists a finite sequence of rationals
$0\leq \valuation^\loc_1<\valuation^\loc_2<\cdots <
\valuation^\loc_k=1$
and a finite sequence of transitions
$\transition_1,\ldots,\transition_k\in\transitions$ such that
\begin{inparaenum}[$(i)$]
\item for all $1\leq i \leq k$, for all
  $\valuation\in (\valuation^\loc_{i-1},\valuation^\loc_{i}]$, either
  $\strat(\loc,\valuation) = (0,\transition_i)$, or
  $\strat(q,\valuation) = (\valuation^\loc_i-\valuation,
  \transition_i)$
  (assuming $\valuation^\loc_0=\min(0,\valuation^\loc_1)$); and
\item if $\valuation^\loc_1> 0$, then
  $\strat(\loc,0) = (\valuation^\loc_1,\transition_1)$.
\end{inparaenum}
We let $\points(\strat)$ be the set of $\valuation^\loc_i$ for all
$\loc$ and $i$, and $\intervals(\strat)$ be the set of all
successive intervals generated by $\points(\strat)$. 
Finally, we let $|\strat|=|\intervals(\strat)|$ be the size of
$\strat$. Intuitively, in an interval
$(\valuation^\loc_{i-1},\valuation^\loc_i]$, $\strat$ always returns
the same move: either to take \emph{immediately} $\transition_i$ or to
wait until the clock reaches the endpoint $\valuation^\loc_i$ and then
take $\transition_i$.

\MinPl, however may require memory to play optimally. Informally, we
will compute optimal \emph{switching} strategies, as introduced in
\cite{BGHM14} (in the untimed setting). A switching strategy is
described by a pair $(\stratmin^1,\stratmin^2)$ of FP-strategies and a
switch threshold $K$, and consists in playing $\stratmin^1$ until the
total accumulated cost of the discrete transitions is below $K$; and
then to \emph{switch} to strategy $\stratmin^2$. The role of
$\stratmin^2$ is to ensure reaching a final location: it is thus a
(classical) attractor strategy. The role of $\stratmin^1$, on the
other hand, is to allow \MinPl to decrease the cost low enough
(possibly by forcing negative cycles) to secure a cost below $K$, and
the computation of $\stratmin^1$ is thus the critical point in the
computation of an optimal switching strategy. To characterise
$\stratmin^1$, we introduce the notion of negative cycle strategy
(NC-strategy). 
Formally, an NC-strategy $\stratmin$ of \MinPl is an FP-strategy such
that for all runs
$\rho = (\loc_1,\valuation)\xrightarrow{c_1} \cdots
\xrightarrow{c_{k-1}} (\loc_k,\valuation')\in \Play{\stratmin}$
with $\loc_1=\loc_k$, and $\valuation,\valuation'$ in the same
interval of $\intervals(\stratmin)$, the sum of prices of
\emph{discrete transitions} is at most $-1$, i.e.,
$\price(\loc_1,\loc_2)+\cdots+\price(\loc_{k-1},\loc_k) \leq -1$. To
characterise the fact that $\stratmin^1$ must allow \MinPl to reach a
cost which is \emph{small enough, without necessarily reaching a
  target state}, we define the \emph{fake value} of an NC-strategy
$\stratmin$ from a configuration $s$ as
$\fakeValue_\game^{\stratmin}(s) = \sup \{\cost{\run} \mid \rho \in
\Play{s,\stratmin}, \rho \textrm{ reaches a target}\}$,
i.e., the value obtained when \emph{ignoring} the $\stratmin$-induced
plays that \emph{do not} reach the target.  Thus, clearly,
$\fakeValue_\game^{\stratmin}(s) \leq \val^{\stratmin}(s)$. We say
that an NC-strategy is \emph{fake-optimal} if its fake value, in every
configuration, is equal to the optimal value of the configuration in
the game. This is justified by the following result whose proof relies
on the switching strategies described before (see a detailed proof in
Appendix~\ref{app:fake-optimality}):

\begin{lemma}\label{lem:fake-optimality}
  If $\val_\game(\loc,\valuation)\neq +\infty$, for all $\loc$ and
  $\valuation$, then for all NC-strategies $\stratmin$, there is a
  strategy $\stratmin'$ such that
  $\Value_\game^{\stratmin'}(s)\leq\fakeValue_\game^{\stratmin}(s)$
  for all configurations $s$. In particular, if $\stratmin$ is a
  fake-optimal NC-strategy, then $\stratmin'$ is an optimal
  (switching) strategy of the \SPTG.
\end{lemma}

Then, an \SPTG is called \emph{finitely optimal} if
\begin{inparaenum}[$(i)$]
\item $\MinPl$ has a fake-optimal NC-strategy;
\item $\MaxPl$ has an optimal FP-strategy; and
\item $\Value_\game(\loc)$ is a cost function, for all locations~$\loc$.
\end{inparaenum} 
The central point in establishing Theorem~\ref{thm:main-result} will
thus be to prove that \textbf{all \SPTG{s} are finitely optimal}, as
this guarantees the existence of well-behaved optimal strategies and
value functions. We will also show that they can be computed in
exponential time.
The proof is by induction on the number of urgent locations of the
\SPTG. In Section~\ref{sec:urgentSPTG}, we address the base case of
\SPTG{s} with urgent locations only (where no time can elapse). Since
these \SPTG{s} are very close to the \emph{untimed} min-cost
reachability games of \cite{BGHM14}, we adapt the algorithm in this
work and obtain the \SolveInstant function
(Algorithm~\ref{algo:value-iteration-fixed}). This function can also
compute $\valgame(\loc,1)$ for all $\loc$ and all games $\game$ (even
with non-urgent locations) since time can not elapse anymore when the
clock has valuation $1$. Next, using the continuity result of
Theorem~\ref{prop:continuity-of-val}, we can detect locations $\loc$
where $\valgame(\loc,\valuation)\in\{+\infty,-\infty\}$, for all
$\valuation\in[0,1]$, and remove them from the game. Finally, in
Section~\ref{sec:solving-sptg} we handle \SPTG{s} with non-urgent
locations by refining the technique of \cite{BouLar06,Rut11} (that
work only on \SPTG{s} with non-negative costs). Compared to
\cite{BouLar06,Rut11}, our algorithm is simpler, being iterative,
instead of recursive.

\section{SPTGs with only urgent locations}
\label{sec:urgentSPTG}

Throughout this section, we consider an $\rightpoint$-\SPTG
$\game= (\LocsMin, \LocsMax, \LocsFin, \LocsUrg, \fgoalvec,
\transitions, \price)$
where all locations are urgent, i.e.,
$\LocsUrg = \LocsMin\cup\LocsMax$. We first explain briefly how we can
compute the value function of the game for a \emph{fixed} clock
valuation $\valuation\in [0,\rightpoint]$ (more precisely, we can
compute the vector $(\val_\game(\loc,\valuation))_{\loc\in\Locs}$).
Since no time can elapse, we can adapt the techniques developed in
\cite{BGHM14} to solve (untimed) \emph{min-cost reachability
  games}. The adaptation consists in taking into account the final
cost functions (see Appendix~\ref{app:urgentSPTG}).  This yields the
function \SolveInstant
(\algorithmcfname~\ref{algo:value-iteration-fixed}), that computes the
vector $(\val_\game(\loc,\valuation))_{\loc\in\Locs}$ for a
fixed~$\valuation$. The results of \cite{BGHM14} also allow us to
compute associated optimal strategies: when
$\val(\loc,\valuation)\notin \{-\infty,+\infty\}$ the optimal strategy
for $\MaxPl$ is memoryless, and the optimal strategy for $\MinPl$ is a
switching strategy $(\stratmin^1,\stratmin^2)$ with a threshold $K$
(as described in the previous section).

\begin{algorithm}[tbp]
  \caption{\texttt{solveInstant}($\game$,$\valuation$)}
  \label{algo:value-iteration-fixed}
  \DontPrintSemicolon%
  \KwIn{$\rightpoint$-\SPTG
    $\game= (\LocsMin, \LocsMax, \LocsFin, \LocsUrg, \fgoalvec, \transitions,
    \price)$, a valuation $\valuation\in [0,\rightpoint]$}%
  \SetKw{value}{\ensuremath{\mathsf{X}}}
  \SetKw{prevvalue}{\ensuremath{\mathsf{X}_{pre}}}
  
  \BlankLine

  \ForEach{$\loc\in\Locs$}{%
    \leIf{$\loc\in \LocsFin$}%
    {$\value(\loc) :=
      \fgoal_\loc(\valuation)$}{$\value(\loc):=+\infty$} %
  }%

  \Repeat{$\value = \prevvalue$}{%
    $\prevvalue := \value$\;%
    \lForEach{$\loc\in\LocsMax$}{$\value(\loc) :=
      \max_{(\loc,\loc')\in\transitions}
      \big(\price(\loc,\loc')+\prevvalue(\loc')\big)$} %
    \lForEach{$\loc\in
      \LocsMin$}{$\value(\loc)
      := \min_{(\loc,\loc')\in\transitions}
      \big(\price(\loc,\loc')+\prevvalue(\loc')\big)$}%
    \lForEach{$\loc\in \Locs$ \emph{such that}
      $\value(\loc) < -(|\Locs|-1)
      \maxPriceTrans-\maxPriceFin$\label{line-infty-RT}\label{line-infty}}%
    {$\value(\loc) := -\infty$\label{line-update}}%
  } %
  \Return{$\value$}
\end{algorithm}

Now let us explain how we can reduce the computation of
$\val_\game(\loc)\colon \valuation\in[0,\rightpoint] \mapsto
\val(\loc,\valuation)$
(for all $\loc$) to a \emph{finite number of calls} to \SolveInstant.
Let $\F_{\game}$ be the set of affine functions over $[0,\rightpoint]$
such that
$\F_{\game} = \{k+\fgoal_\loc\mid \loc\in\LocsFin \land k\in {\cal
  I}\}$,
where
${\cal I}=[-(|\Locs|-1)\maxPriceTrans,|\Locs|\maxPriceTrans]\cap\Z$.
Observe that $\F_\game$ has cardinality $2|\Locs|^2\maxPriceTrans$,
i.e., pseudo-polynomial in the size of $\game$. From \cite{BGHM14}, we
conclude that the functions in $\F_\game$ are sufficient to
characterise $\val_\game$, in the following sense: for all
$\loc\in \Locs$ and $\valuation\in [0,\rightpoint]$ such that
$\val(\loc,\valuation)\notin\{-\infty,+\infty\}$, there is
$f\in\F_{\game}$ with $\val(\loc,\valuation)=f(\valuation)$ (see
Lemma~\ref{lem:for-all-val-there-is-f-in-F},
Appendix~\ref{app:urgentSPTG} for the details). Using the continuity
of $\val_\game$ (Theorem~\ref{prop:continuity-of-val}), we show that
all the cutpoints of $\val_\game$ are intersections of functions
from~$\F_{\game}$, i.e., belong to the set of \emph{possible
  cutpoints}
$\posscp_\game =\{\valuation\in [0,\rightpoint]\mid \exists
f_1,f_2\in\F_{\game}\quad f_1\neq f_2\land
f_1(\valuation)=f_2(\valuation)\}$.
Observe that $\posscp_\game$ contains at most
$|\F_{\game}|^2=4|\LocsFin|^4(\maxPriceTrans)^2$ points (also a
pseudo-polynomial in the size of~$\game$) since all functions in
$\F_{\game}$ are affine, and can thus intersect at most once with
every other function.  Moreover, $\posscp_\game\subseteq \Q$, since
all functions of $\F_{\game}$ take rational values in $0$ and
$\rightpoint\in\Q$. Thus, for all $\loc$, $\val_\game(\loc)$ is a cost
function (with cutpoints in $\posscp_\game$ and pieces from
$\F_\game$).  Since $\val_\game(\loc)$ is a piecewise affine function,
we can characterise it completely by computing only its value on its
cutpoints. Hence, we can reconstruct $\val_\game(\loc)$ by calling
\SolveInstant on each rational valuation
$\valuation \in \posscp_\game$.  From the optimal strategies computed
along \SolveInstant \cite{BGHM14}, we can also reconstruct a
fake-optimal NC-strategy for $\MinPl$ and an optimal FP-strategy for
$\MaxPl$, hence:
\begin{proposition}\label{prop:baseCase}
  Every $r$-\SPTG $\game$ with only urgent locations is finitely
  optimal. Moreover, for all locations~$\loc$, the piecewise affine
  function $\val_\game(\loc)$ has cutpoints in $\posscp_{\game}$ of
  cardinality $4|\LocsFin|^4(\maxPriceTrans)^2$, pseudo-polynomial in
  the size of $\game$.
\end{proposition}

\section{Solving general SPTGs}
\label{sec:solving-sptg}

In this section, we consider \SPTG{s} with possibly non-urgent
locations. We first prove that all such \SPTG{s} are finitely
optimal. Then, we introduce Algorithm~\ref{alg:solve} to compute
optimal values and strategies of \SPTG{s}. To the best of our
knowledge, this is the first algorithm to solve \SPTG{s} with
arbitrary weights.  Throughout the section, we fix an \SPTG
$\game = (\LocsMin, \LocsMax, \LocsFin, \LocsUrg, \fgoalvec,
\transitions, \price)$
with possibly non-urgent locations.  Before presenting our core
contributions, let us explain how we can detect locations with
infinite values. As already argued, we can compute $\val(\loc,1)$ for
all $\loc$ assuming all locations are urgent, since time can not
elapse anymore when the clock has valuation $1$. This can be done with
\SolveInstant. Then, by continuity, $\val(\loc,1)=+\infty$
(respectively, $\val(\loc,1)=-\infty$), for some $\loc$ if and only if
$\val(\loc,\valuation)=+\infty$ (respectively,
$\val(\loc,\valuation)=-\infty$) for all $\valuation \in[0,1]$. We
remove from the game all locations with infinite value without
changing the values of other locations (as justified
in~\cite{BGHM14}). Thus, we henceforth assume that
$\val(\loc,\valuation)\in\R$ for all $(\loc,\valuation)$.

\subparagraph*{The $\game_{\Locs',r}$ construction.} To prove finite
optimality of \SPTG{s} and to establish correctness of our algorithm,
we rely in both cases on a construction that consists in decomposing
$\game$ into a sequence of \SPTG{s} with \emph{more urgent
  locations}. Intuitively, a game with more urgent locations is easier
to solve since it is closer to an untimed game (in particular, when
all locations are urgent, we can apply the techniques of
Section~\ref{sec:urgentSPTG}). More precisely, given a set~$\Locs'$ of
non-urgent locations, and a valuation $r_0\in [0,1]$, we will define a
(possibly infinite) sequence of valuations $1=r_0> r_1>\cdots$ and a
sequence $\game_{\Locs',r_0}$, $\game_{\Locs',r_1},\ldots$ of \SPTG{s}
such that
\begin{inparaenum}[$(i)$]
\item all locations of $\game$ are also present in each
  $\game_{\Locs', r_i}$, except that the locations of $\Locs'$ are now
  urgent; and
\item for all $i\geq 0$, the value function of $\game_{\Locs',r_i}$ is
  equal to $\val_\game$ on the interval $[r_{i+1},r_i]$. Hence, we can
  re-construct $\val_\game$ by assembling
  well-chosen parts of the values functions of the
  $\game_{\Locs',r_i}$ (assuming $\inf_i r_i=0$).
\end{inparaenum}
This basic result will be exploited in two directions. First, we prove
by induction on the number of urgent locations that all \SPTG{s} are
finitely optimal, by re-constructing $\val_\game$ (as well as optimal
strategies) as a $\opcf$-concatenation of the value functions of a
finite sequence of \SPTG{s} with one more urgent locations. The base
case, with only urgent locations, is solved by
Proposition~\ref{prop:baseCase}. This construction suggests a
\emph{recursive} algorithm in the spirit of \cite{BouLar06,Rut11} (for
non-negative prices). Second, we show that this recursion can be
\emph{avoided} (see Algorithm~\ref{alg:solve}). Instead of turning
locations urgent one at a time, this algorithm makes them all urgent
and computes directly the sequence of \SPTG{s} with only urgent
locations. Its proof of correctness relies on the finite optimality of
\SPTG{s} and, again, on our basic result linking the values functions
of $\game$ and games $\game_{\Locs',r_i}$.

Let us formalise these constructions. Let $\game$ be an \SPTG, let
$r\in[0,1]$ be an endpoint, and let $\vec x = (x_\loc)_{\loc\in\locs}$
be a vector of rational values. Then, $\Waiting(\game,r,\vec x)$ is an
$r$-\SPTG in which both players may now decide, in all non-urgent
locations $\loc$, to \emph{wait} until the clock takes value $r$, and
then to stop the game, adding the cost $x_\loc$ to the current cost of
the play. Formally,
$\Waiting(\game,r,\vec x) = (\LocsMin,\LocsMax,\LocsFin',\LocsUrg,
\fgoalvec', T', \price')$
is such that
$\LocsFin' = \LocsFin \uplus \{\loc^f \mid \loc\in \locs\setminus
\LocsUrg\}$;
for all $\loc' \in \LocsFin$ and $\valuation\in[0,r]$,
$\fgoal'_{\loc'}(\valuation) =\fgoal_{\loc'}(\valuation)$, for all
$\loc\in\locs\setminus \LocsUrg$,
$\fgoal'_{\loc^f}(\valuation)=(r-\valuation)\cdot
\price(\loc)+x_\loc$;
$T'=T\cup \{(\loc, [0,r], \bot,\loc^f)\mid \loc\in \locs\setminus
\LocsUrg\}$;
for all $\transition\in T'$,
$\price'(\transition)=\price(\transition)$ if $\transition\in T$, and
$\price'(\transition) =0$ otherwise. Then, we let
$\game_r=\Waiting\big(\game,r,(\Value_\game(\loc,r))_{\loc\in\Locs}\big)$,
i.e., the game obtained thanks to $\Waiting$ by letting $\vec x$ be
the value of $\game$ in $r$.  One can check that this first
transformation does not alter the value of the game, for valuations
before $r$:
$\Value_{\game}(\loc,\valuation) = \Value_{\game_r}(\loc,\valuation)$
for all $\valuation\leq r$.

Next, we make locations urgent. For a set
$\Locs'\subseteq \Locs\setminus\LocsUrg$ of non-urgent locations, we
let $\game_{\Locs',r}$ be the \SPTG obtained from $\game_r$ by making
urgent every location $\loc$ of $\Locs'$. Observe that, although all
locations $\loc\in \Locs'$ are now urgent in $\game_{\Locs',r}$, their
clones $\loc^f$ allow the players to wait until $r$. When $\Locs'$ is
a singleton $\{\loc\}$, we write $\game_{\loc,r}$ instead of
$\game_{\{\loc\},r}$.
While the construction of $\game_r$ does not change the value of the
game, introducing urgent locations \emph{does}. Yet, we can
characterise an interval $[a,r]$ on which the value functions of
$\hame=\game_{L',r}$ and $\hame^+=\game_{L'\cup\{\ell\},r}$ coincide,
as stated by the next proposition. The interval $[a,r]$ depends on the
\emph{slopes} of the pieces of $\val_{\hame^+}$ as depicted in
Figure~\ref{fig:slopeBigger}: for each location $\ell$ of $\MinPl$,
the slopes of the pieces of $\val_{\hame^+}$ contained in $[a,r]$
should be $\leq-\price(\loc)$ (and $\geq-\price(\loc)$ when $\loc$
belongs to $\MaxPl$). It is proved by lifting optimal strategies of
$\hame^+$ into $\hame$, and strongly relies on the determinacy result
of Theorem~\ref{thm:determined}:

\begin{proposition}\label{lem:SameValue}
  Let $0\leq a < r \leq 1$, $\Locs'\subseteq \Locs\setminus\LocsUrg$
  and $\loc\notin \Locs'\cup\LocsUrg$ a non-urgent location of \MinPl
  (respectively, \MaxPl). Assume that $\game_{\Locs'\cup\{\loc\},r}$
  is finitely optimal, and for all
  $a\leq \valuation_1<\valuation_2 \leq r$
  \begin{equation}
    \frac{\Value_{\game_{\Locs'\cup\{\loc\},r}}(\loc,\valuation_2) -
      \Value_{\game_{\Locs'\cup\{\loc\},r}}(\loc,\valuation_1)} 
    {\valuation_2-\valuation_1} \geq
    -\price(\loc) \quad (\textrm{respectively, }\leq
    -\price(\loc))\,.\label{eq:SlopeBigger} 
  \end{equation}
  Then, for all $\valuation\in [a,r]$ and $\loc'\in \locs$,
  $\Value_{\game_{\Locs'\cup\{\loc\},r}}(\loc',\valuation) =
  \Value_{\game_{\Locs',r}}(\loc',\valuation)$.
  Furthermore, fake-optimal NC-strategies and optimal FP-strategies in
  $\game_{\Locs'\cup\{\loc\},r}$ are also fake-optimal and optimal
  over $[a,r]$ in $\game_{\Locs',r}$.
\end{proposition}

Given an \SPTG $\game$ and some \emph{finitely optimal}
$\game_{\Locs',r}$, we now characterise precisely the left endpoint of
the maximal interval ending in $r$ where the value functions of
$\game$ and $\game_{L',r}$ coincide, with the operator
$\Next_{\Locs'}\colon (0,1]\to [0,1]$ (or simply $\Next$, if $\Locs'$
is clear) defined as:
\[\Next_{\Locs'}(r)=  \sup \{ r'\leq r \mid \forall \loc\in \Locs \
\forall \valuation \in [r',r] \
\Value_{\game_{\Locs',r}}(\loc,\valuation) =
\Value_{\game}(\loc,\valuation)\}\,.\]
By %
continuity of the value (Theorem~\ref{prop:continuity-of-val}), this
supremum exists and
$\Value_{\game}(\loc,\Next_{\Locs'}(r)) =
\Value_{\game_{\Locs',r}}(\loc,\Next_{\Locs'}(r))$.
Moreover, $\Value_{\game}(\loc)$ is a cost function on $[\Next(r),r]$,
since $\game_{\Locs',r}$ is finitely optimal. However, this definition
of $\Next(r)$ is semantical. Yet, building on the ideas of
Proposition~\ref{lem:SameValue}, we can effectively compute
$\Next(r)$, given $\Value_{\game_{\Locs',r}}$. We claim that
$\Next_{\Locs'}(r)$ is the \emph{minimal valuation} such that for all
locations $\loc\in \Locs'\cap\LocsMin$ (respectively,
$\loc \in \Locs'\cap\LocsMax$), the slopes of the affine sections of
the cost function $\Value_{\game_{\Locs',r}}(\loc)$ on $[\Next(r),r]$
are at least (at most) $-\price(\loc)$ (see
Lemma~\ref{lem:operatorNext} in appendix). Hence, $\Next(r)$ can be
obtained (see Figure~\ref{fig:HowToConstructri}), by inspecting
iteratively, for all $\loc$ of $\MinPl$ (respectively, $\MaxPl$), the
slopes of $\Value_{\game_{\Locs',r}}(\loc)$, by decreasing valuations,
until we find a piece with a slope $>-\price(\loc)$ (respectively,
$<-\price(\loc)$). This enumeration of the slopes is effective as
$\Value_{\game_{\Locs',r}}$ has finitely many pieces, by
hypothesis. Moreover, this guarantees that $\Next(r)<r$. Thus, one can
reconstruct $\val_\game$ on $[\inf_i r_i,r_0]$ from the value
functions of the (potentially infinite) sequence of games
$\game_{\Locs',r_0}$, $\game_{\Locs',r_1},\ldots$ where
$r_{i+1}=\Next(r_i)$ for all $i$ such that $r_i>0$, for all possible
choices of non-urgent locations $\Locs'$. Next, we will define two
different ways of choosing $\Locs'$: the former to prove finite
optimality of all \SPTG{s}, the latter to obtain an algorithm to solve
them.

\begin{figure}[tbp]
  \begin{minipage}[t]{.47\linewidth}
    \centering
    \footnotesize
    \begin{tikzpicture}[scale=0.6]
      \node[above] at (0,5) {$\Value_{\game_{\loc,r}}(\loc,\valuation)$};
      \node[below] at (7,0) {$\valuation$};
      \draw[->] (-0.3,0) -- (7,0);
      \draw[->] (0,-0.3) -- (0,5);
      \draw[dotted] (6,0) -- (6,5);
      \draw[dotted] (3,0) -- (3,5);

      \node[below] at (3,0) {$a$};
      \node[below] at (6,0) {$r$};

      \path[name path=courbe,draw] (0,0.5)  .. controls (0.5,7) and ( 1.5,0) .. (6,4); 
      \path[name path = x] (3.5,0) -- (3.5,7);
      \path[name path = xprime] (5.5,0) -- (5.5,7);

      \draw[thick, name intersections={of=courbe and x,by={c}},name intersections={of=courbe and xprime,by={d}}]  (c)--(d);

      \node at (c) {$\bullet$};
      \node at (d) {$\bullet$};

      \coordinate (p) at ($(d)-(2,-0.5)$); 

      \draw[dashed] (c) -- ($(c)+(2,-0.5)$);

      \draw[dotted] (3.5,0) -- (c);
      \draw[dotted] (5.5,0) -- (d);
      \node[below] at (3.5,0) {$\valuation_1$};
      \node[below] at (5.5,0) {$\valuation_2$};

      \coordinate(cy) at ($(c)-(3.5,0)$);
      \coordinate(dy) at ($(d)-(5.5,0)$);

      \draw[dotted] (cy) -- (c);
      \draw[dotted] (dy) -- (d);

      \node[left] at (cy) {$\Value_{\game_{\loc,r}}(\loc,\valuation_1)$};
      \node[left] at (dy) {$\Value_{\game_{\loc,r}}(\loc,\valuation_2)$};

    \end{tikzpicture}
    \caption{The condition~\eqref{eq:SlopeBigger} (in the case
      $\Locs'=\emptyset$ and $\loc\in \LocsMin$): graphically, it means
      that the slope between any two points of the plot in $[a,r]$
      (represented with a thick line) is greater than or equal to
      $-\price(\loc)$ (represented with dashed line).}
    \label{fig:slopeBigger}
  \end{minipage} \ \ \ 
  \begin{minipage}[t]{.47\linewidth}
    \centering
    \footnotesize
    \begin{tikzpicture}[scale=0.6]
      \node[above] at (0,5) {$\Value_{\game_{\locMin,r}}(\locMin,\valuation)$};
      \node[below] at (7,0) {$\valuation$};
      \draw[->] (-0.3,0) -- (7,0);
      \draw[->] (0,-0.3) -- (0,5);
      \draw[dotted] (6,0) -- (6,5);
      \draw[dotted] (3,0) -- (3,5);
      \draw[dotted] (0,4) -- (7,4);
      \node[left] at (0,4) {$\Value_\game(\locMin,r)$};
      \node[below] at (3,0) {$\Next(r)$};
      \node[below] at (6,0) {$r$};

      \draw (6,4) -- (5,3) -- (4,2.7) -- (3.7,2) -- (3,1) -- (2,1.7)
      -- (1.5,1) -- (1,1.3) -- (0.5,0.5) -- (0,0.5);    
      \begin{scope}[xshift=5cm, yshift=3cm]
        \draw[dashed] (0,0) -- (0.5,0.125);
      \end{scope}
      \begin{scope}[xshift=4cm, yshift=2.5cm]
        \draw[dashed] (0,0) -- (0.5,0.125);
      \end{scope}
      \begin{scope}[xshift=3.7cm, yshift=2cm]
        \draw[dashed] (0,0) -- (0.5,0.125);
      \end{scope}
      \begin{scope}[xshift=3cm, yshift=1cm]
        \draw[dashed] (0,0) -- (0.5,0.125);
      \end{scope}
      \begin{scope}[xshift=2cm, yshift=1.7cm]
        \draw[dashed] (0,0) -- (0.5,0.125);
      \end{scope}
      \begin{scope}[xshift=1.5cm, yshift=1cm]
        \draw[dashed] (0,0) -- (0.5,0.125);
      \end{scope}
      \begin{scope}[xshift=1cm, yshift=1.3cm]
        \draw[dashed] (0,0) -- (0.5,0.125);
      \end{scope}
      \begin{scope}[xshift=0.5cm, yshift=0.5cm]
        \draw[dashed] (0,0) -- (0.5,0.125);
      \end{scope}
      \begin{scope}[xshift=0cm, yshift=0.5cm]
        \draw[dashed] (0,0) -- (0.5,0.125);
      \end{scope}
    \end{tikzpicture}
    \caption{In this example $\Locs'=\{\locMin\}$
      and $\locMin\in \LocsMin$. $\Next(r)$ is the
      leftmost point such that all slopes on its right are smaller
      than or equal to $-\price(\locMin)$ in the graph of
      $\Value_{\game_{\locMin,r}}(\locMin,\valuation)$. Dashed lines
      have slope $-\price(\locMin)$.}
    \label{fig:HowToConstructri}
  \end{minipage}
  
\end{figure}

\subparagraph*{\SPTG{s} are finitely optimal.} To prove finite
optimality of all \SPTG{s} we reason by induction on the number of
non-urgent locations and instantiate the previous results to the case
where $L'=\{\locMin\}$ where $\locMin$ is a non-urgent location of
\emph{minimum price-rate} (i.e., for all $\loc\in \Locs$,
$\price(\locMin)\leq \price(\loc)$). Given $r_0\in [0,1]$, we let
$r_0> r_1> \cdots$ be the decreasing sequence of valuations such that
$r_i=\Next_{\locMin}(r_{i-1})$ for all $i>0$. As explained before, we
will build $\val_\game$ on $[\inf_i r_i, r_0]$ from the value
functions of games $\game_{\locMin,r_i}$. Assuming finite optimality
of those games, this will prove that $\game$ is finitely optimal
\emph{under the condition} that $r_0> r_1> \cdots$ eventually stops,
i.e., $r_i=0$ for some $i$. This property is given by the next lemma,
which ensures that, for all $i$, the owner of $\locMin$ has a strictly
better strategy in configuration $(\locMin,r_{i+1})$ than waiting
until $r_i$ in location $\locMin$. 

\begin{lemma}\label{lem:strictlySmaller}\label{lem:stationarysequence-locMin}
  If $\game_{\locMin,r_i}$ is finitely optimal for all $i\geq 0$, then
  \begin{inparaenum}[$(i)$]
  \item if $\locMin\in\LocsMin$ (respectively, $\LocsMax$),
    $\Val_\game(\locMin,r_{i+1}) <
    \Val_\game(\locMin,r_i)+(r_i-r_{i+1})\price(\locMin)$
    (respectively,
    $\Val_\game(\locMin,r_{i+1}) >
    \Val_\game(\locMin,r_i)+(r_i-r_{i+1})\price(\locMin)$),
    for all $i$; and
  \item there is $i\leq |\F_\game|^2+2$ such that $r_i=0$.
  \end{inparaenum}
\end{lemma}

By iterating this construction, we make all locations urgent
iteratively, and obtain:

\begin{proposition}\label{prop:ExpCutpoints}
  Every \SPTG $\game$ is finitely optimal and for all locations
  $\loc$, $\Value_\game(\loc)$ has at most
  $O\left ( (\maxPriceTrans|\Locs|^2)^{2|\locs|+2}\right)$ cutpoints.
\end{proposition}
\begin{proof} 
  As announced, we show by induction on $n\geq 0$ that every $r$-\SPTG
  $\game$ with $n$ non-urgent locations is finitely optimal, and that
  the number of cutpoints of $\Value_\game(\loc)$ is at most
  $O\left ( (\maxPriceTrans (|\LocsFin|+n^2))^{2n+2}\right)$, which
  suffices to show the above bound, since
  $|\LocsFin|+n^2\leq |\Locs|^2$.

  The base case $n=0$ is given by
  Proposition~\ref{prop:baseCase}. Now, assume that $\game$ has at
  least one non-urgent location, and consider $\locMin$ one with
  minimum price. By induction hypothesis, all $r'$-\SPTG{s}
  $\game_{\locMin,r'}$ are finitely optimal for all $r'\in [0,r]$.
  Let $r_0> r_1> \cdots$ be the decreasing sequence defined by $r_0=r$
  and $r_i=\Next_{\locMin}(r_{i-1})$ for all $i\geq 1$. By
  Lemma~\ref{lem:stationarysequence-locMin}, there exists
  $j\leq |\F_\game|^2+2$ such that $r_j = 0$. Moreover, for all
  $0<i\leq j$, $\Value_\game=\Value_{\game_{\locMin,r_{i-1}}}$ on
  $[r_{i},r_{i-1}]$ by definition of $r_i = \Next_{\locMin}(r_{i-1})$,
  so that $\Value_\game(\loc)$ is a cost function on this interval,
  for all $\loc$, and the number of cutpoints on this interval is
  bounded by
  $O\left ( (\maxPriceTrans
    (|\LocsFin|+(n-1)^2+n))^{2(n-1)+2}\right)=O\left ( (\maxPriceTrans
    (|\LocsFin|+n^2))^{2(n-1)+2}\right)$
  by induction hypothesis (notice that maximal transition prices are
  the same in $\game$ and $\game_{\locMin,r_{i-1}}$, but that we add
  $n$ more final locations in $\game_{\locMin,r_{i-1}}$). Adding the
  cutpoint $1$, summing over $i$ from $0$ to $j\leq |\F_\game|^2+2$,
  and observing that $|\F_\game|\leq 2 \maxPriceTrans |\LocsFin|$, we
  bound the number of cutpoints of $\Value_\game(\loc)$ by
  $O\left ( (\maxPriceTrans (|\LocsFin|+n^2))^{2n+2}\right)$.
  Finally, we can reconstruct fake-optimal and optimal strategies in
  $\game$ from the from fake-optimal and optimal strategies of
  $\game_{\locMin,r_i}$.
\end{proof}

\subparagraph*{Computing the value functions.}  The finite optimality
of \SPTG{s} allows us to compute the value functions. The proof of
Proposition~\ref{prop:ExpCutpoints} suggests a \emph{recursive}
algorithm to do so: from an \SPTG $\game$ with minimal non-urgent
location $\locMin$, solve recursively $\game_{\locMin,1}$,
$\game_{\locMin,\Next(1)}$, $\game_{\locMin,\Next(\Next(1))}$,
\textit{etc}.\ handling the base case where all locations are urgent
with Algorithm~\ref{algo:value-iteration-fixed}. While our results
above show that this is correct and terminates, we propose instead to
solve---without the need for recursion---the sequence of games
$\game_{\Locs\setminus\LocsUrg,1}$,
$\game_{\Locs\setminus\LocsUrg,\Next(1)}, \ldots$ i.e., \emph{making
  all locations urgent at once}. Again, the arguments given above
prove that this scheme is \emph{correct}, but the key argument of
Lemma~\ref{lem:stationarysequence-locMin} that ensures
\emph{termination} can not be applied in this case. Instead, we rely
on the following lemma, stating, that there will be at least one
cutpoint of $\val_\game$ in each interval $[\Next(r),r]$. Observe that
this lemma relies on the fact that $\game$ is finitely optimal, hence
the need to first prove this fact independently with the sequence
$\game_{\locMin,1}$, $\game_{\locMin,\Next(1)}$,
$\game_{\locMin,\Next(\Next(1))}$,\ldots Termination then follows from
the fact that $\val_\game$ has finitely many cutpoints by finite
optimality.

\begin{lemma}\label{lem:r_2-r_1-r_0}
  Let $r_0\in(0,1]$ such that $\game_{\Locs',r_0}$ is finitely
  optimal. Suppose that $r_1=\Next_{\Locs'}(r_0)>0$, and let
  $r_2=\Next_{\Locs'}(r_1)$. There exists $r'\in [r_2,r_1)$ and
  $\loc\in \Locs'$ such that
  \begin{inparaenum}[(i)][$(i)$]
  \item $\Value_{\game}(\loc)$ is affine on $[r',r_1]$, of slope
    equal to $-\price(\loc)$, and
  \item  $\Value_{\game}(\loc,r_1) \neq \Value_{\game}(\loc,r_0) +
    \price(\loc) (r_0-r_1)$.
  \end{inparaenum}
  As a consequence, $\Value_\game(\loc)$ has a cutpoint in
  $[r_1,r_0)$.
\end{lemma}

Algorithm~\ref{alg:solve} implements these ideas. Each iteration of
the \textbf{while} loop computes a new game in the sequence
$\game_{\Locs\setminus\LocsUrg,1}$,
$\game_{\Locs\setminus\LocsUrg,\Next(1)},\ldots$ described above;
solves it thanks to \SolveInstant; and thus computes a new portion of
$\val_\game$ on an interval on the left of the current point
$r\in [0,1]$. More precisely, the vector
$(\val_\game(\loc,1))_{\loc\in \Locs}$ is first computed in
line~\ref{alg:init}. Then, the algorithm enters the {\bf while} loop,
and the game $\game'$ obtained when reaching
line~\ref{alg:game-constructed} is
$\game_{\Locs\setminus\LocsUrg,1}$. Then, the algorithm enters the
{\bf repeat} loop to analyse this game. Instead of building the whole
value function of $\game'$, Algorithm~\ref{alg:solve} builds only the
parts of $\val_{\game'}$ that coincide with $\val_{\game}$. It
proceeds by enumerating the possible cutpoints $a$ of $\val_{\game'}$,
starting in~$r$, by decreasing valuations (line~\ref{alg:cutpoint}),
and computes the value of $\val_{\game'}$ in each cutpoint thanks to
\SolveInstant (line~\ref{alg:call-si}), which yields a new piece of
$\val_{\game'}$. Then, the {\bf if} in line~\ref{line:begin} checks
whether this new piece coincides with $\val_\game$, using the
condition given by Proposition~\ref{lem:SameValue}. If it is the case,
the piece of $\val_{\game'}$ is added to $f_\loc$
(line~\ref{alg:concat}); \textbf{repeat} is stopped otherwise. When
exiting the \textbf{repeat} loop, variable $b$ has value
$\Next(1)$. Hence, at the next iteration of the \textbf{while} loop,
$\game'=\game_{\Locs\setminus\LocsUrg,\Next(1)}$ when reaching
line~\ref{alg:game-constructed}. By continuing this reasoning
inductively, one concludes that the successive iterations of the
\textbf{while} loop compute the sequence
$\game_{\Locs\setminus\LocsUrg,1}$,
$\game_{\Locs\setminus\LocsUrg,\Next(1)},\ldots$ as announced, and
rebuilds $\val_\game$ from them. Termination in exponential time is
ensured by Lemma~\ref{lem:r_2-r_1-r_0}: each iteration of the {\bf
  while} loop discovers at least one new cutpoint of $\val_\game$, and
there are at most exponentially many (note that a tighter bound on
this number of cutpoints would entail a better complexity of our
algorithm).

\begin{example} Let us briefly sketch the execution of
  Algorithm~\ref{alg:solve} on the \SPTG in
  Figure~\ref{fig:ex-ptg2}. During the first iteration of the {\bf
    while} loop, the algorithm computes the correct value functions
  until the cutpoint $\frac 3 4$: in the $repeat$ loop, at first
  $a= 9/10$ but the slope in $\loc_1$ is smaller than the slope that
  would be granted by waiting, as depicted in
  Figure~\ref{fig:ex-ptg2}. Then, $a=3/4$ where the algorithm gives a
  slope of value $-16$ in $\loc_2$ while the cost of this location of
  \MaxPl is $-14$. During the first iteration of the {\bf while} loop,
  the inner {\bf repeat} loop thus ends with $r=3/4$. The next
  iterations of the {\bf while} loop end with $r=\frac 1 2$ (because
  $\loc_1$ does not pass the test in line~\ref{line:begin});
  $r=\frac 1 4$ (because of $\loc_2$) and finally with $r=0$, giving
  us the value functions on the entire interval $[0,1]$. All value
  functions are in Figure~\ref{fig:val_sptg} in the appendix.
\end{example}

\begin{algorithm}[tbp]
  \caption{\texttt{solve}($\game$)}\label{alg:solve}
  \KwIn{\SPTG
    $\game=(\LocsMin, \LocsMax, \LocsFin, \LocsUrg, \fgoalvec,
    \transitions, \price)$} %
  \DontPrintSemicolon%

  $\vec f = (f_\loc)_{\loc\in \locs} :=
  \SolveInstant(\game,1)$\tcc*[r]{$f_\loc\colon \{1\}\to \Rbar$}\label{alg:init}%
  $r := 1$\;%
  \While(\tcc*[f]{Invariant: $f_\loc\colon [r,1] \to \Rbar$}){$0<r$}{%
    $\game' := \Waiting(\game,r,\vec f(r))$ \tcc*[r]{$r$-\SPTG
      $\game' = (\LocsMin,\LocsMax,\LocsFin',\LocsUrg', \fgoalvec',
      T', \price')$}%
    $\LocsUrg' := \LocsUrg' \cup \locs$\tcc*[r]{every location is made
      urgent}%
    $b := r$\label{alg:game-constructed}\;%
    \Repeat(\tcc*[f]{Invariant:$f_\loc\colon [b,1]\to \Rbar$}){ $b=0$
      or $stop$}%
    {%
      $a := \max (\posscp_{\game'}\cap [0,b))$\label{alg:cutpoint}\;%
      $\vec x = (x_\loc)_{\loc\in \locs}
      :=\SolveInstant(\game',a)$\tcc*[r]{$x_\loc
        = \val_{\game'}(\loc,a)$}\label{alg:call-si}%
      \If{$\forall \loc\in\LocsMin \;
        \frac{f_\loc(b)-x_\loc}{b-a}\leq-\price(\loc)\land \forall
        \loc\in\LocsMax \;
        \frac{f_\loc(b)-x_\loc}{b-a}\geq-\price(\loc)$\label{line:begin}}%
      {%
        \lForEach{$\loc\in\locs$}
        {$f_\loc := \big(\valuation\in[a,b] \mapsto f_\loc(b) +
          (\valuation-b)\frac{f_\loc(b)-x_\loc}{b-a}\big) \rhd
          f_\loc$}\label{alg:concat}%
        $b:=a$\,; $stop := false$ }%
      \lElse{$stop := true$\label{line:end}}
    }%
    $r:= b$%
  }%
  \Return{$\vec f$}%
\end{algorithm}

\section{Beyond \SPTG{s}}
\label{sec:solve-PTGs}

In \cite{BouLar06,Rut11,DueIbs13}, \emph{general} \PTG{s} with
\emph{non-negative prices} are solved by reducing them to a finite
sequence of \SPTG{s}, by eliminating guards and resets. It is thus
natural to try and adapt these techniques to our general case, in
which case Algorithm~\ref{alg:solve} would allow us to solve
\emph{general \PTG{s} with arbitrary costs}. Let us explain why it is
not (completely) the case. The technique used to remove guards from
\PTG{s} consists in enhancing the locations with regions while keeping
an equivalent game. This technique \emph{can} be adapted to arbitrary
weights, see Appendix~\ref{app:raptg} for a proof adapted from
\cite[Lemma~4.6]{DueIbs13}.

The technique to handle resets, however, consists in \emph{bounding}
the number of clock resets that can occur in any play following an
optimal strategy of \MinPl or \MaxPl. Then, the \PTG can be
\emph{unfolded} into a \emph{reset-acyclic} \PTG with the same
value. By reset-acyclic, we mean that no cycle in the configuration
graph visits a transition with a reset. This reset-acyclic \PTG can be
decomposed into a finite number of components that contain no reset
and are linked by transitions with resets. These components can be
solved iteratively, from the bottom to the top, turning them into
\SPTG{s}.  Thus, if we \emph{assume} that the \PTG{s} we are given as
input \emph{are} reset-acyclic, we can solve them in \emph{exponential
  time}, and show that their value functions are cost functions with
at most exponentially many cutpoints, using our techniques (see
Appendix~\ref{app:raptg}). Unfortunately, the arguments to bound the
number of resets do not hold for arbitrary costs, as shown by the \PTG
in \figurename~\ref{fig:ex-ptgrr}. We claim that $\val(\loc_0)=0$;
that \MinPl has no optimal strategy, but a family of
$\varepsilon$-optimal strategies $\stratmin^\varepsilon$ each with
value~$\varepsilon$; and that each $\stratmin^\varepsilon$ requires
\emph{memory whose size depends on $\varepsilon$} and might
\emph{yield a play visiting at least $1/\varepsilon$ times the reset}
between $\loc_0$ and $\loc_1$ (hence the number of resets can not be
bounded). For all $\varepsilon>0$, $\stratmin^\varepsilon$ consists
in: waiting $1-\varepsilon$ time units in $\loc_0$, then going to
$\loc_1$ during the $\lceil 1/\varepsilon\rceil$ first visits to
$\loc_0$; and to go directly to $\loc_f$ afterwards. Against
$\stratmin^\varepsilon$, $\MaxPl$ has two possible choices:
\begin{inparaenum}[$(i)$]
\item either wait $0$ time unit in $\loc_1$, wait $\varepsilon$ time
  units in $\loc_2$, then reach $\loc_f$; or
\item wait $\varepsilon$ time unit in $\loc_1$ then force the cycle by
  going back to $\loc_0$ and wait for $\MinPl$'s next move.
\end{inparaenum}
Thus, all plays according to $\stratmin^\varepsilon$ will visit a
sequence of locations which is either of the form
$\loc_0 (\loc_1 \loc_0)^k\loc_1\loc_2\loc_f$, with
$0\leq k <\lceil 1/\varepsilon\rceil$; or of the form
$\loc_0 (\loc_1
\loc_0)^{\left\lceil\frac{1}{\varepsilon}\right\rceil}\loc_f$.
In the former case, the cost of the play will be
$-k\varepsilon+0+\varepsilon=-(k-1)\varepsilon\leq\varepsilon$; in the
latter, $-\varepsilon(\lceil 1/\varepsilon\rceil)+1\leq 0$. This shows
that $\val(\loc_0)=0$, but there is no optimal strategy as none of
these strategies allow one to guarantee a cost of $0$ (neither does
the strategy that waits $1$ time unit in $\loc_0$).

\begin{figure}[tbp] 
  \centering 
  \begin{tikzpicture}[minimum size=5mm, node distance = 3cm] 
    \everymath{\footnotesize}

    \node[draw,circle, label=left:$0$] at (0,0) (q0) {\makebox[0pt][c]{$\loc_0$}};
    \node[draw,rectangle,right of=q0,label=above:$-1$] (q1) {\makebox[0pt][c]{$\loc_1$}};
    \node[draw,rectangle, right of=q1, label=above:$1$] (q2) {\makebox[0pt][c]{$\loc_2$}};
    \node[draw,circle,accepting, right of=q2] (qf) {\makebox[0pt][c]{$\loc_f$}};

    \path[draw,arrows=-latex'] 
    (q2) edge node[above] {$x=1$} (qf) 
    (q1) edge node[above] {$x\leq 1$} (q2) 
    (q0) edge[bend left=10]   node[above] {$x\leq 1$} (q1)
    (q1) edge[bend left=10] node[below] {$x=1, x:=0$} (q0);

    \draw[arrows=-latex',rounded corners] (q0.north) -- ++(0,.5) -- node[above] {$1$}
    ($ (qf.north)+(0,.5) $) -- (qf) ;

  \end{tikzpicture}
  \caption{A \PTG where the number of resets in optimal plays can not
    be bounded a priori.} \label{fig:ex-ptgrr}
\end{figure}
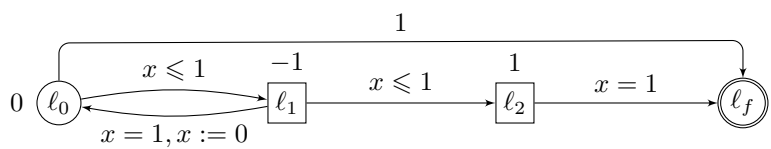

However, we may apply the result on reset-acyclic \PTG{s} to obtain:
\begin{theorem} 
  The value functions of all one-clock \PTG{s} are cost functions with
  at most exponentially many cutpoints.
\end{theorem} 
\begin{proof} 
  Let $\game$ be a one-clock \PTG. Let us replace all transitions
  $(\loc,g,\top,\loc')$ resetting the clock by $(\loc,g,\bot,\loc'')$,
  where $\loc''$ is a new final location with
  $\fgoal_{\loc''}=\val_\game(\loc,0)$---observe that
  $\val_\game(\loc,0)$ exists even if we can not compute it, so this
  transformation is well-defined. This yields a reset-acyclic \PTG
  $\game'$ such that $\val_{\game'}=\val_\game$.
\end{proof}

\bibliographystyle{plain}

\newpage
\changepage{3cm}{3cm}{-1.5cm}{-1.5cm}{}{-1cm}{}{}{}


\appendix

\section{Existence and continuity of the value functions: proof of
  Theorem~\ref{prop:continuity-of-val}\label{app:continuity-of-val}}

We start with the proof of determinacy. For all $k\in \R$, define
$\textit{Threshold}(\game,r)$ as the \emph{qualitative} game which is
played like $\game$, and only the objective of $\MinPl$ is altered (in
order to make it qualitative): now $\MinPl$ wins a play if and only if
the cost of the play is $\leq k$. Further, let $P(k)$ be the set of
prefixes of runs ending in a final vertex and whose cost is less than
or equal to $k$. Then the set of winning plays for $\MinPl$ in this
game is $S=\bigcup_{\rho\in P(k)} \textit{Cone}(\rho)$ where
$\textit{Cone}(\rho)$ denotes the set of plays having $\rho$ as a
prefix. The set $S$ is an open set in the topology induced by
cones. In~\cite{GS1953}, it is shown that in any game whose set of
winning plays is an open set is \emph{determined}, i.e. one of the two
players has a winning strategy. Therefore
$\textit{Threshold}(\game,k)$ is determined for all $k$.

Now let us prove that $\lvalgame = \uvalgame$. First, recall that, by
definition of $\lvalgame$ and $\uvalgame$:
\begin{align}
  \lvalgame(c)&\leq \uvalgame(c) \label{eq:2}
\end{align}
for all configurations $c$. Fix a configuration $c$. We consider
several cases:
\begin{enumerate}
\item First assume that $\lvalgame(c) \in \R$. By definition,
  for all $k>\lvalgame(c)$ and all strategies $\sigma_\MaxPl$,
  $\valgs{\game}{\sigma_\MaxPl}(c)<k$.
  %
  Hence, for all $k>\lvalgame(c)$, $\MaxPl$ has no
  winning strategy in the game
  $\textit{Threshold}(\game,k)$. Therefore, by determinacy of this
  game, $\MinPl$ has a winning strategy. Equivalently, for all $k>\lvalgame(c)$,
  there exists $\sigma^{k}_\MinPl$ such that
  $\valgs{\game}{\sigma^{k}_\MinPl}(c)\leq k$. This implies that:
  \begin{align}
    \uvalgame(c)&\leq \lvalgame(c)\label{eq:1}
  \end{align}
  Hence, by~\eqref{eq:1} and \eqref{eq:2} we conclude that:
  $\lvalgame(c) = \uvalgame(c)$ when these values are finite.

\item In the case where $\lvalgame(c) = +\infty$, we conclude,
  by~\eqref{eq:2} that $\uvalgame(c)=+\infty$ too.

\item Finally, in the case where $\lvalgame(c) = -\infty$ then for all
  $k$, $\MaxPl$ has no winning strategy for
  $\textit{Threshold}(\game,k)$. Therefore, by determinacy, $\MinPl$
  has a winning strategy $\sigma^k_\MinPl$ in
  $\textit{Threshold}(\game,k)$. Thus, for all $k$:
  $\uvalgame\leq \valgs{\game}{\sigma^{k}_\MinPl}(c) \leq k$, and:
  $\uvalgame=-\infty$.
\end{enumerate}

\medskip We then turn to the proof of continuity. Therefore, our goal
is to show that for every location $\loc$, region $r\in \reggame$ and
valuations $\valuation$ and $\valuation'$ in $r$,
\[|\val(\loc,\valuation)-\val(\loc,\valuation')| \leq
\maxPriceLoc|\valuation-\valuation'|.\]

This is equivalent to showing
\[\val(\loc,\valuation)\leq\val(\loc,\valuation') +
\maxPriceLoc|\valuation-\valuation'|\quad \textrm{and}\quad
\val(\loc,\valuation')\leq\val(\loc,\valuation) +
\maxPriceLoc|\valuation-\valuation'|\,.\]
As those two equations are symmetric with respect to $\valuation$
and $\valuation'$, we only have to show either of them. We will thus
focus on the latter, which, by using the upper value, can be
reformulated as: for all strategies $\stratmin$ of \MinPl, there
exists a strategy $\stratmin'$ such that
$\valgs{}{\stratmin'}(\loc,\valuation')\leq
\valgs{}{\stratmin}(\loc,\valuation)+\maxPriceLoc|\valuation-\valuation'|$.
Note that this last equation is equivalent to say that there exists
a function $g$ mapping plays $\rho'$ from $(\loc,\valuation')$,
consistent with $\stratmin'$ (i.e., such that
$\rho'=\Play{(\loc,\valuation'),\stratmin',\stratmax}$ for some
strategy $\stratmax$ of $\MaxPl$) to plays from $(\loc,\valuation)$,
consistent with $\stratmin$, such that
\[\cost{\run'}\leq
\cost{g(\run')}+\maxPriceLoc|\valuation-\valuation'|\,.\]

Let $r \in \reggame$, $\valuation,\valuation'\in r$ and $\stratmin$
be a strategy of \MinPl. We define $\stratmin'$ and $g$ by induction
on the size of their arguments; more precisely, we define
$\stratmin'(\run'_1)$ and $g(\run'_2)$ by induction on $k$, for all
plays $\run'_1$ and $\run'_2$ from $(\loc,\valuation')$, consistent
with $\stratmin'$ of size $k-1$ and $k$, respectively. We also show
during this induction that for each play
$\run'=(\loc_1,\valuation'_1)\xrightarrow{c'_1}\cdots\xrightarrow{c'_{k-1}}
(\loc_k,\valuation'_k)$
from $(\loc,\valuation')$, consistent with $\stratmin')$, if we let
$(\loc_1,\valuation_1)\xrightarrow{c_1}\cdots\xrightarrow{c_{\ell-1}}
(\loc_k,\valuation_{\ell}) = g(\run')$:
\begin{enumerate}[$(i)$]
\item $\run'$ and $g(\run')$ have the same length, i.e.,
  $\len{\run}=\ell=k=\len{\run'}$, \item for every
  $i\in\{1,\ldots,k\}$, $\valuation_i$ and $\valuation'_i$ are in
  the same region, i.e., there exists a region $r'\in \reggame$ such
  that $\valuation_i\in r'$ and $\valuation'_i\in r'$,
\item $|\valuation_k-\valuation'_k|\leq |\valuation-\valuation'|$,
\item
  $\cost{\run'} \leq \cost{g(\run')} +
  \maxPriceLoc(|\valuation-\valuation'|
  -|\valuation_n-\valuation'_n|)$.
\end{enumerate}
Notice that no property is required on the strategy $\stratmin'$ for
finite plays that do not start in $(\loc,\valuation')$.

If $k=1$, as there is no play of length $0$, nothing has to be done
to define $\stratmin'$. Moreover, in that case,
$\run' = (\loc,\valuation')$ and $g(\run')=(\loc,\valuation)$. Both
plays have size~$1$, $\valuation$ and $\valuation'$ are in the same
region by hypothesis of the lemma, and
$\cost{\run'} = \cost{g(\run')} =0$, therefore all four properties
are true.

Let us suppose now that the construction is done for a given
$k\geq 1$, and perform it for $k+1$. We start with the construction
of $\stratmin'$. To that extent, consider a play
$\run'=(\loc_1,\valuation'_1)\xrightarrow{c'_1}\cdots\xrightarrow{c'_{k-1}}
(\loc_k,\valuation'_k)$
from $(\loc,\valuation')$, consistent with $\stratmin'$ such that
$\loc_{k}$ is a location of player $\MinPl$. Let $t$ and
$\transition$ be the choice of delay and transition made by
$\stratmin$ on $g(\run')$, i.e.,
$\stratmin(g(\run'))=(t,\transition)$. Then, we define
$\stratmin'(\run')=(t',\transition)$ where
$t' = \max(0, \valuation_{k}+t - \valuation'_{k})$.  The delay $t'$
respects the guard of transition $\transition$ since either
$\valuation_{k}+t = \valuation'_{k}+t'$ or
$\valuation_{k}\leq \valuation_{k}+t\leq \valuation'_{k}$, in which
case $\valuation'_{k}$ is in the same region as $\valuation_{k}+t$
since $\valuation_k$ and $\valuation'_k$ are in the same
region. This is illustrated in Figure~\ref{fig:deltaPrime}.

\begin{figure}[tbp]
  \centering
  \begin{tikzpicture}
    \begin{scope}
      \draw[->,dashed] (-0.5,0) -- (2.5,0);  
      \node  at (0,-0.3) {$\valuation'_{k}$};
      \node (vkp) at (0,0) {$\bullet$};
      \node  at (1,0.3) {$\valuation_{k}$};
      \node (vk) at (1,0) {$\bullet$};
      \node (vkdelta) at (2,0) {$\bullet$};

      \node at (1,-1.5) {(a)};

      \draw[->] (vk) to[bend left] node[midway, above] {$t$} (vkdelta);
      \draw[->] (vkp) to[bend right]  node[midway, below]  {$t'$} (vkdelta);
    \end{scope}

    \begin{scope}[xshift = 4cm]
      \draw[->,dashed] (-0.5,0) -- (2.5,0);  
      \node  at (0,-0.3) {$\valuation_{k}$};
      \node (vk) at (0,0) {$\bullet$};
      \node  at (1,0.3) {$\valuation'_{k}$};
      \node (vkp) at (1,0) {$\bullet$};
      \node (vkdelta) at (2,0) {$\bullet$};

      \node at (1,-1.5) {(b)};

      \draw[->] (vk) to[bend right] node[midway, below] {$t$} (vkdelta);
      \draw[->] (vkp) to[bend left]  node[midway,above]  {$t'$} (vkdelta);
    \end{scope}

    \begin{scope}[xshift = 8cm]
      \draw[->,dashed] (-0.5,0) -- (2.5,0);  
      \node  at (0,-0.3) {$\valuation_{k}$};
      \node (vk) at (0,0) {$\bullet$};
      \node (vkdelta) at (1,0) {$\bullet$};
      \node  at (2,-0.3) {$\valuation'_{k}$};
      \node (vkp) at (2,0) {$\bullet$};

      \node at (1,-1.5) {(c)};

      \draw[->] (vk) to[bend right] node[midway, below] {$t$} (vkdelta);
      \draw[->]  (vkp) edge[loop above] node[midway, above] {$t'$} (vkp);
    \end{scope}

  \end{tikzpicture}

  \caption{The definition of $t'$ when (a)
    $\valuation'_k\leq \valuation_k$, (b)
    $\valuation_k < \valuation'_k < \valuation_k+t$, (c)
    $\valuation_k < \valuation_k+t < \valuation'_k $.}
  \label{fig:deltaPrime}
\end{figure}

We now build the mapping $g$. Let
$\run'=(\loc_1,\valuation'_1)\xrightarrow{c'_1}\cdots\xrightarrow{c'_{k}}
(\loc_{k+1},\valuation'_{k+1})$
be a play from $(\loc,\valuation')$ consistent with $\stratmin'$ and
$\tilde{\run}'
=(\loc_1,\valuation'_1)\xrightarrow{c'_1}\cdots\xrightarrow{c'_{k-1}}
(\loc_k,\valuation'_k)$
its prefix of size $k$. Let $(t', \transition)$ be the delay and
transition taken after $\tilde{\run}'$. Using the construction of
$g$ over plays of length $k$ by induction, the play
$g(\tilde{\run}') =
(\loc_1,\valuation_1)\xrightarrow{c_1}\cdots\xrightarrow{c_{k-1}}
(\loc_k,\valuation_k)$
(with $(\loc_1,\valuation_1)=(\loc,\valuation)$) verifies properties
(i), (ii) and (iii). If $\loc_k$ is a location of \MinPl and
$\stratmin(g(\tilde{\run}')) = (t,\transition)$, then
$g(\run')=g(\tilde{\run}')\xrightarrow{c_{k}}
(\loc_{k+1},\valuation_{k+1})$
is obtained by applying those choices on $g(\tilde{\run}')$. If
$\loc_k$ is a location of \MaxPl, the last valuation
$\valuation_{k+1}$ of $g(\run')$ is rather obtained by choosing
action $(t,\transition)$ verifying
$t = \max (0, \valuation'_{k}+t'-\valuation_k)$. Note that
transition $\transition$ is allowed since both $\valuation_k+t$ and
$\valuation'_k + t'$ are in the same region (for similar reasons as
above).

By induction hypothesis $|\tilde{\run}'| = |g(\tilde{\run}')|$, thus
(i) holds, i.e., $|\run'|=|g(\run')|$. Moreover, $\valuation_{k+1}$
and $\valuation'_{k+1}$ are also in the same region as either they
are equal to $\valuation_k+t$ and $\valuation'_k + t'$,
respectively, or $\transition$ contains a reset in which case
$\valuation_{k+1}=\valuation'_{k+1}=0$ which proves (ii). To prove
(iii), notice that we always have either
$\valuation_k+t = \valuation'_k + t'$ or
$\valuation_k\leq \valuation_k+t \leq
\valuation'_k=\valuation'_k+t'$
or
$\valuation'_k\leq \valuation'_k+t \leq
\valuation_k=\valuation_k+t$.
In all of these possibilities, we have
$|(\valuation_k+t) - (\valuation'_k+t)|\leq
|\valuation_k-\valuation'_k|$.
By noticing again that either $\valuation_{k+1}=\valuation_k+t$
and $\valuation'_{k+1}=\valuation'_k + t'$, or $\transition$
contains a reset in which case
$\valuation_{k+1}=\valuation'_{k+1}=0$, we conclude the proof of
(iii). We finally check property (iv). In both cases: 
\begin{align*}
  \cost{\run'} &= \cost{\tilde{\run}'}+\price(\transition)+t' \price(\loc_k) \\
               &\leq
                 \cost{g(\tilde{\run}')}+\maxPriceLoc(|\valuation-\valuation'|
                 -|\valuation_k-\valuation'_k|)+\price(\transition)+t'
                 \price(\loc_k)\\ 
               &=\cost{g(\run')}
                 +(t'-t)\price(\loc_k)+\maxPriceLoc(|\valuation-\valuation'|
                 -|\valuation_k-\valuation'_k|)\,. 
\end{align*}
If $\transition$ contains no reset, let us prove that
\begin{equation}
  |t'-t|=|\valuation_k-\valuation'_k| -
  |\valuation'_{k+1}-\valuation_{k+1}|\,.\label{eq:t'-t}
\end{equation}
Indeed, since $t' = \valuation'_{k+1}-\valuation'_{k}$ and
$t = \valuation_{k+1}-\valuation_{k}$, we have
$|t'-t|
=|\valuation'_{k+1}-\valuation'_{k}-(\valuation_{k+1}-\valuation_{k})|$.
Then, two cases are possible: either
$t' = \max(0, \valuation_{k}+t - \valuation'_{k})$ or
$t = \max (0, \valuation'_{k}+t'-\valuation_k)$. So we
have three different possibilities:
\begin{itemize}
\item if $t'+\valuation'_k=t +\valuation_k$,
  $\valuation'_{k+1}=\valuation_{k+1}$, thus
  $|t'-t|= |\valuation_k-\valuation'_k| =
  |\valuation_k-\valuation'_k| -
  |\valuation'_{k+1}-\valuation_{k+1}|$.
\item if $t=0$, then
  $\valuation_k=\valuation_{k+1}\geq \valuation'_{k+1}\geq
  \valuation'_k$,
  thus
  $|\valuation'_{k+1}-\valuation'_{k}-(\valuation_{k+1}-\valuation_{k})|=
  \valuation'_{k+1}-\valuation'_k = (\valuation_{k}-\valuation'_k) -
  (\valuation_k-\valuation'_{k+1}) = |\valuation_k-\valuation'_k| -
  |\valuation'_{k+1}-\valuation_{k+1}|$.
\item if $t'=0$, then
  $\valuation'_k=\valuation'_{k+1}\geq \valuation_{k+1}\geq
  \valuation_k$,
  thus
  $|\valuation'_{k+1}-\valuation'_{k}-(\valuation_{k+1}-\valuation_{k})|=
  \valuation_{k+1}-\valuation_k = (\valuation'_{k}-\valuation_k) -
  (\valuation'_k-\valuation_{k+1}) = |\valuation_k-\valuation'_k| -
  |\valuation'_{k+1}-\valuation_{k+1}|$.
\end{itemize}
If $\transition$ contains a reset, then
$\valuation'_{k+1}=\valuation_{k+1}$. If
$t' = \valuation_{k}+t - \valuation'_{k}$, we have that
$|t'-t| = |\valuation_k-\valuation'_k| $. Otherwise, either $t=0$
and $t' \leq \valuation_k-\valuation'_k$, or $t'=0$ and
$t \leq \valuation'_k-\valuation_k$.

In all cases, we have proved \eqref{eq:t'-t}. Coupled with the fact
that $|P(\loc_k)|\leq \maxPriceLoc$, we conclude that:
\[\cost{\run'}\leq \cost{g(\run')}
+\maxPriceLoc(|\valuation-\valuation'|
-|\valuation_{k+1}-\valuation'_{k+1}|)\,.\]

Now that $\stratmin'$ and $g$ are defined (noticing that $g$ is
stable by prefix, we extend naturally its definition to infinite
plays), notice that for all play $\run'$ from $(\loc,\valuation')$
consistent with $\stratmin'$, either $\run'$ does not reach a final
location and its cost is $+\infty$, but in this case $g(\run')$ has
also cost $+\infty$; or $\run'$ is finite. In this case let
$\valuation'_k$ be the clock valuation of its last configuration,
and $\valuation_k$ be the clock valuation of the last configuration
of $g(\run')$. Combining $(iii)$ and $(iv)$ we have
$\cost{\run'}\leq
\cost{g(\run')}+\maxPriceLoc|\valuation-\valuation'|$
which concludes the proof.

\section{Non-continuity of the value function with more than one
  clock}
\label{app:pas-continu2}

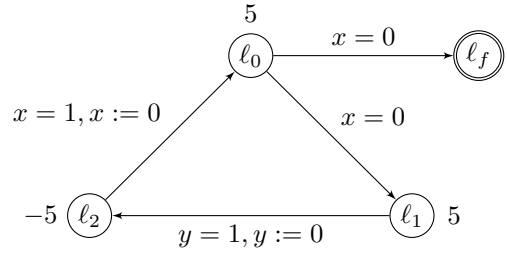
\begin{figure}[tbp]
  \begin{center}
    \begin{tikzpicture}[node distance=3cm,minimum size=5mm]
      \node[draw,circle,label=above:$5$] (q0)
      {\makebox[0pt][c]{$\loc_0$}};%
      \node[draw,circle,below left of=q0, label=left:$-5$] (c2)
      {\makebox[0pt][c]{$\loc_2$}};%
      \node[draw,circle,right of=q0,accepting] (q1)
      {\makebox[0pt][c]{$\loc_f$}};%
      \node[draw,circle,below right of=q0,label=right:$5$] (c1)
      {\makebox[0pt][c]{$\loc_1$}};
      
      \path[draw,arrows=-latex'] (q0) edge node[above] {$x=0$} (q1)%
      (q0) edge node[above right] {$x=0$} (c1) %
      (c1) edge node[below] {$y=1, y:=0$} (c2) %
      (c2) edge node[above left] {$x=1, x:=0$} (q0);
    \end{tikzpicture}
  \end{center}
  \caption{A PTG with 2 clocks whose value function is not
    continuous inside a region}
  \label{fig:pas-continu2}
\end{figure}
Let us consider the example
in \figurename~\ref{fig:pas-continu2} (that we describe informally
since we did not properly define games with multiple clocks), with
clocks $x$ and $y$.  One can easily check that, starting from a
configuration $(\loc_0,0,0.5)$ in location $\loc_0$ and where $x=0$
and $y= 0.5$, the following cycle can be taken:
$(\loc_0, 0, 0.5)\xrightarrow{0,\transition_0, 0}(\loc_1,0,
0.5)\xrightarrow{0.5, \transition_1, 2.5}(\loc_2,0.5,0)
\xrightarrow{0.5, \transition_2, -2.5}(\loc_0, 0, 0.5)$, %
where $\transition_0$, $\transition_1$ and $\transition_2$ denote
respectively the transitions from $\loc_0$ to $\loc_1$; from $\loc_1$
to $\loc_2$; and from $\loc_2$ to $\loc_0$.  Observe that the cost of
this cycle is null, and that no other delays can be played, hence
$\uval(\loc_0,0,0.5)=0$.  However, starting from a configuration
$(\loc_0, 0, 0.6)$, and following the same path, yields the cycle %
$(\loc_0, 0, 0.6)\xrightarrow{0,e_0, 0}(\loc_1, 0,
0.6)\xrightarrow{0.4, e_1, 2}(\loc_2,0.4,0) \xrightarrow{0.6, e_2,
  -3}(\loc_0, 0, 0.6)$ %
with cost $-1$. Hence, $\uval(\loc_0,0,0.6)=-\infty$, and the function
is not continuous although both valuations $(0, 0.5)$ and $(0, 0.6)$
are in the same region. Observe that this holds even for priced timed
\emph{automata}, since our example requires only one player.

\section{Memory is required for \MinPl to play optimally}
\label{app:ex-JF}

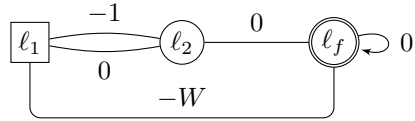
\begin{figure}[tbp]
\begin{center}
  \begin{tikzpicture}[node distance=2cm,minimum size=5mm,>=latex']
      \node[draw,rectangle](1){\makebox[0mm][c]{$\loc_1$}}; 
      \node[draw,circle](2)[right of=1]{\makebox[0mm][c]{$\loc_2$}}; 
      \node[draw,circle,accepting](3)[right of=2]{\makebox[0mm][c]{$\loc_f$}};
      
      \draw[rounded corners]
      (1) -- (0, -1) -- node[above]{$-W$} (4, -1) -- (3);

      \path 
      (1) edge[bend left=10] node[above]{$-1$} (2) 
      (2) edge[bend left=10] node[below]{$0$} (1)
          edge node[above]{$0$} (3)
      (3) edge[loop right] node[right]{$0$} (3);
    \end{tikzpicture}
\end{center}
\caption{An \SPTG where \MinPl needs memory to play optimally}
\label{fig:Weighted-game}
\end{figure}

As an example, consider the \SPTG of
\figurename~\ref{fig:Weighted-game}, where $W$ is a positive integer,
and every location has price-rate~$0$: hence, this game can be seen as
an (untimed) min-cost reachability game as studied in \cite{BGHM14},
where it has been initially studied. We claim that the values of
locations $\loc_1$ and $\loc_2$ are both $-W$. Indeed, consider the
following strategy for $\MinPl$: during each of the first $W$ visits
to $\loc_2$ (if any), go to $\loc_1$; else, go to $\loc_f$. Clearly,
this strategy ensures that the final location $\loc_f$ will eventually
be reached, and that either
\begin{inparaenum}[$(i)$]
\item transition $(\loc_1,\loc_3)$ (with weight~$-W$) will eventually
  be traversed; or
\item transition $(\loc_1,\loc_2)$ (with weight~$-1$) will be traversed at
  least~$W$ times.
\end{inparaenum}
Hence, in all plays following this strategy, the cost will be at
most~$-W$. This strategy allows $\MinPl$ to secure~$-W$, but he can
not ensure a lower cost, since $\MaxPl$ always has the opportunity
to take the transition $(\loc_1,\loc_f)$ (with weight~$-W$) instead of
cycling between $\loc_1$ and $\loc_2$. Hence, $\MaxPl$'s optimal
choice is to follow the transition $(\loc_1,\loc_f)$ as soon as
$\loc_1$ is reached, securing a cost of $-W$. The $\MinPl$ strategy
we have just given is optimal, and there is \emph{no optimal
  memoryless strategy} for $\MinPl$. Indeed, always playing
$(\loc_2,\loc_f)$ does not ensure a cost at most $-W$; and, always
playing $(\loc_2,\loc_1)$ does not guarantee to reach the target, and
this strategy has thus value $+\infty$.

\section{Fake-optimality: proof of Lemma~\ref{lem:fake-optimality}}
\label{app:fake-optimality}

First of all, notice that all finite plays $\run \in \Play{\stratmin}$
with all clock valuations in the same interval $I$ of
$\intervals(\strat)$ verify
$\cost{\run} \leq |I|\maxPriceLoc + |\locs|\maxPriceTrans
-|\run|/|\locs|$.
Indeed, the cost of $\run$ is the sum of the cost generated by staying
in locations, which is bounded by $|I|\maxPriceLoc$, and the cost of
the transitions. One can extract at least $|\run|/|\locs|$ cycles with
transition prices as most $-1$ (by definition of an NC-strategy), and
what remains is of size at most $|\locs|$, ensuring that the
transition cost is bounded by $|\locs|\maxPriceTrans -|\run|/|\locs|$.

Then, by splitting runs among intervals of $\intervals(\stratmin)$, we
can easily obtain that all finite plays $\run\in\Play{\stratmin}$
verify
$\cost{\run} \leq \maxPriceLoc +
(2|\stratmin|-1)\times|\locs|\maxPriceTrans
-(|\run|-|\stratmin|)/|\locs|$.
Indeed, letting $I_1,I_2,\ldots,I_{k}$ the interval of
$\intervals(\stratmin)$ visited during $\run$ (with
$k\leq |\stratmin|$), one can split $\run$ into $k$ runs
$\run = \run_1 \xrightarrow{c_1} \run_2 \xrightarrow{c_2}
\cdots\run_{k}$
such that in $\run_i$ all clock values are in $I_i$ (remember that
\SPTG{s} contain no reset transitions). By the previous inequality, we
have
$\cost{\run_i} \leq |I_i| \maxPriceLoc + |\locs|\maxPriceTrans
-|\run_i|/|\locs|$.
Thus, also splitting costs $c_i$ with respect to discrete cost and
cost of delaying, we obtain
$\cost{\run} = \sum_{i=1}^{k}\cost{\run_i} + \sum_{i=1}^{k-1} c_i \leq
(2|\stratmin|-1)\times |\locs|\maxPriceTrans + \maxPriceLoc -
(|\run|-|\stratmin|)/|\locs$,
since
$|\run|\leq \sum_i |\run_i| + k \leq \sum_i |\run_i|+|\stratmin|$ and
$\sum_i |I_i| \leq 1$.

We now turn to the proof of the lemma. To that extent, we suppose
known an attractor strategy for \MinPl, i.e., a strategy that ensures
to reach a final location: it exists thanks to the hypothesis on the
finiteness of the values. From every configuration, it reaches a final
location with a cost bounded above by a given constant~$M$. Notice
first that, with the hypothesis that no configuration has a value
$-\infty$ in the \SPTG we consider, it is not possible that
$\fakeValue_\game^{\stratmin}(s) = -\infty$ for a configuration $s$
(i.e., that no runs of $\Play{s,\stratmin}$ reach the target). Indeed,
consider the strategy $\stratmin'$ obtained by playing $\stratmin$
until having computed a cost bounded above by a fixed integer
$N\in\Z$, in which case we switch to the attractor strategy. By the
previous inequality, the switch is sure to happen since the right term
tends to $-\infty$ when the length of $\run$ tends to $\infty$. Then,
we know that the value guaranteed by $\stratmin'$ is at most $N$,
implying that the optimal value $\Value(s)$ is $-\infty$, which
contradicts the hypothesis. Then, to prove the result of the lemma,
consider the strategy $\stratmin'$ obtained by playing $\stratmin$
until having computed a cost bounded above by the finite value
$\fakeValue_\game^{\stratmin}(s)-M$, in which case we switch to the
attractor strategy. Once again, the switch is sure to happen, implying
that every play conforming to $\stratmin$ reaches the target:
moreover, the cost of such a play is necessarily at most
$\fakeValue^{\stratmin}(s)$ by construction. Then, we directly obtain
that
$\Value_\game^{\stratmin'}(s) \leq \fakeValue_\game^{\stratmin}(s)$.

\section{\SPTG{s} with only urgent locations: extended version of
  Section~\ref{sec:urgentSPTG}}
\label{app:urgentSPTG}

We rely on the proofs of \cite{BGHM14} that can easily be
adapted in our case, even though we must give the whole explanation
here, knowing that prices coming from goal functions can be rational,
and hence do not strictly fall in the framework of \cite{BGHM14}.

Since all locations in $\game$ are urgent, we may extract from a play
$\rho=(\loc_0,\valuation)\xrightarrow{c_0}
(\loc_1,\valuation)\xrightarrow{c_1}\cdots$
the clock valuations, as well as prices
$c_i = \price(\loc_i,\loc_{i+1})$, hence denoting plays by their
sequence of locations $\loc_0\loc_1\cdots$. The cost of this play is
$\cost{\rho}=+\infty$ if $\loc_k\not\in\LocsFin$ for all $k\geq 0$;
and
$\cost{\rho}=\sum_{i=0}^{k-1}\price(\loc_i,\loc_{i+1}) +
\fgoal_{\loc_k}(\valuation)$
if $k$ is the least position such that $\loc_k\in\LocsFin$.

\subsection{Computing the value for a particular valuation}
\label{sec:particular-valuation} 

Let us show how to compute the vector
$\val_\valuation = (\val(\loc,\valuation))_{\loc\in\Locs}$, for a
given $\valuation \in[0,\rightpoint]$, in terms of a sequence of
values. Following the arguments of \cite{BGHM14}, we first observe
that locations $\loc$ with values $\val_\valuation(\loc)=+\infty$ and
$\val_\valuation(\loc)=-\infty$ can be pre-computed (using
respectively attractor and mean-payoff techniques) and removed from
the game without changing the values of the other nodes. Then, because
of the particular structure of the game $\game$ (where a real cost is
paid only on the target location, all other prices being integers),
for all plays $\rho$, $\cost{\rho}$ is a value from the set
$\Zstar =\Z+\{\fgoal_\loc(\valuation)\mid \loc\in \LocsFin\}$. We
further define $\Zstarinf = \Zstar\cup\{+\infty\}$. Clearly, $\Zstar$
contains at most $|\LocsFin|$ values between two consecutive integers,
i.e.,
\begin{equation}
  \forall i\in \Z\quad |[i,i+1]\cap\Zstar| \leq
  |\LocsFin| \label{eq:at-least-Qfin} 
\end{equation}

Then, we define an operator
$\operator\colon (\Zstarinf)^\Locs \to (\Zstarinf)^\Locs$ mapping
every vector $\vec x = (x_\loc)_{\loc\in\Locs}$ of $(\Zstarinf)^\Locs$
to $\operator(\vec x) = (\operator(\vec x)_\loc)_{\loc\in\Locs}$
defined by
\[\operator(\vec x)_\loc =
\begin{cases}
  \fgoal_\loc(\valuation) &\textrm{if } \loc\in\LocsFin\\
  \displaystyle{\max_{(\loc,\loc')\in \transitions}}
  \big(\price(\loc,\loc')+x_{\loc'}\big)
  &\textrm{if } \loc\in \LocsMax\\
  \displaystyle{\min_{(\loc,\loc')\in\transitions}}
  \big(\price(\loc,\loc')+x_{\loc'}\big) &\textrm{if } \loc\in
  \LocsMin\,.
\end{cases}\]
We will obtain $\val_{\valuation}$ as the limit of the sequence
$(\vec x^{(i)})_{i\geq 0}$ defined by $x^{(0)}_\loc= +\infty$ if
$\loc\not\in\LocsFin$, and $x^{(0)}_\loc=\fgoal_\loc(\valuation)$ if
$\loc\in\LocsFin$, and then $\vec x^{(i)}=\operator(\vec x^{(i-1)})$
for $i\geq 0$. 

The intuition behind is that \emph{$x_i$ is the value of the game
  (when the clock takes value $\valuation$) if we impose that $\MinPl$
  must reach the target within $i$ steps} (and get a payoff of
$+\infty$ if it fails to do so). Formally, for a play
$\rho=\loc_0 \loc_1 \cdots$, we let $\costbound{i}(\rho)=\cost{\rho}$
if $\loc_k\in\LocsFin$ for some $k\leq i$, and
$\costbound{i}(\rho)=+\infty$ otherwise. We further let
$\bupval{i}(\loc)=\inf_{\minstrategy} \sup_{\maxstrategy}
\costbound{i}(\outcomes((\loc,\valuation),\maxstrategy,\minstrategy))$
(where $\maxstrategy$ and $\minstrategy$ are respectively strategies
of $\MaxPl$ and $\MinPl$). Lemma~1 of \cite{BGHM14} allows us to
easily obtain that
\begin{lemma}
  For all $i\geq 0$, and $\loc\in\Locs$:
  $\vec x^{(i)}_\loc=\bupval{i}(\loc)$.
\end{lemma}

Now, let us study how the sequence $(\bupval{i})_{i\geq 0}$ behaves
and converges to the finite values of the game. Using again the same
arguments as in \cite{BGHM14} ($\operator$ is a monotonic and
Scott-continuous operator over the complete lattice
$(\Zstarinf)^\Locs$, etc), the sequence $(\bupval{i})_{i\geq 0}$
converges towards the greatest fixed point of $\operator$. Let us now
show that $\val_\valuation$ is actually this greatest fixed
point. First, Corollary~1 of \cite{BGHM14} can be adapted to obtain


\begin{lemma}\label{lem:after-n-steps-no-infty}
  For all $\loc\in\Locs$:
  $\bupval{|\Locs|}(\loc)\leq |\Locs| \maxPriceTrans+\maxPriceFin\,$.
\end{lemma}

The next step is to show that the values that can be computed along
the sequence (still assuming that $\val_{\valuation}(\loc)$ is finite
for all $\loc$) are taken from a finite set:
\begin{lemma}
  For all $i\geq 0$ and for all $\loc\in\Locs$:
  \[\bupval{|\Locs|+i}(\loc) \in \possval_{\valuation}=
  [-(|\Locs|-1) \maxPriceTrans-\maxPriceFin, |\Locs|
  \maxPriceTrans+\maxPriceFin] \cap \Zstar\]
  where $\possval_{\valuation}$ has cardinality bounded by
  $|\LocsFin|\times \big( (2|\Locs|-1)\maxPriceTrans + 2\maxPriceFin +
  1\big)$.
\end{lemma}
\begin{proof}
  Following the proof of \cite[Lemma~3]{BGHM14}, it is easy to show
  that if $\MinPl$ can secure, from some vertex $\loc$, a cost less
  than $-(|\Locs|-1) \maxPriceTrans-\maxPriceFin$, i.e.,
  $\val(\loc,\valuation)<-(|\Locs|-1) \maxPriceTrans-\maxPriceFin$, then it
  can secure an arbitrarily small cost from that configuration, i.e.,
  $\val(\loc,\valuation)=-\infty$, which contradicts our hypothesis
  that the value is finite.


  Hence, for all $i\geq 0$, for all $\loc$:
  $\bupval{i}(\loc)\geq \val_{\valuation}(\loc)> -(|\Locs|-1)
  \maxPriceTrans-\maxPriceFin$.
  By Lemma~\ref{lem:after-n-steps-no-infty} and since the sequence is
  non-increasing, we conclude that, for all $i\geq 0$ and for all
  $\loc\in\Locs$:
  \[-(| \Locs|-1) \maxPriceTrans-\maxPriceFin <
  \bupval{|\Locs|+i}(\loc)\leq |\Locs|
  \maxPriceTrans+\maxPriceFin\,.\]
  Since all $\bupval{|\Locs|+i}(\loc)$ are also in $\Zstar$, we
  conclude that $\bupval{|\Locs|+i}(\loc)\in\possval_{\valuation}$ for
  all $i\geq 0$. The upper bound on the size of
  $\possval_{\valuation}$ is established by
  \eqref{eq:at-least-Qfin}. 
\end{proof}

This allows us to bound the number of iterations needed for the
sequence to stabilise. The worst case is where all locations are
assigned a value bounded below by
$-(|\Locs|-1) \maxPriceTrans-\maxPriceFin$ from the highest possible
value where all vertices are assigned a value bounded above by
$|\Locs| \maxPriceTrans+\maxPriceFin$, which is itself reached after
$|\Locs|$ steps. Hence:
\begin{corollary}
  The sequence $(\bupval{i})_{i\geq 0}$ stabilises after a number of
  steps at most
  $|\LocsFin|\times |\Locs|\times \big( (2|\Locs|-1)\maxPriceTrans +
  2\maxPriceFin + 1\big)+ |\Locs|$.
\end{corollary}

Finally, the proofs of \cite[Lemma~4 and Corollary~2]{BGHM14} allow us
to conclude that this sequence converges towards the value
$\val_{\valuation}$ of the game (when all values are finite), which
proves that the value iteration scheme of
\algorithmcfname~\ref{algo:value-iteration-fixed} computes exactly
$\val_{\valuation}$ for all $\valuation\in [0,\rightpoint]$.  Indeed,
this algorithm also works when some values are not finite. As a
corollary, we obtain a characterisation of the possible values of
$\game$:
\begin{corollary}\label{cor:possible-values}
  For all $\rightpoint$-\SPTG $\game$ with only urgent locations, for
  all location $\loc\in\Locs$ and valuation
  $\valuation\in [0,\rightpoint]$, $\valgs{\game}{}(\loc,\valuation)$
  is contained in the set
  $\possval_{\valuation}\cup\{-\infty,+\infty\}$ of cardinal
  $O(\mathrm{poly}(|\Locs|,\maxPriceTrans,\maxPriceFin))$,
  pseudo-polynomial with respect to the size of~$\game$.
\end{corollary}

Finally, Section 3.4 and 3.5 of \cite {BGHM14} explain how to compute
simultaneously optimal strategies for both players. In our context,
this allows us to obtain for every valuation
$\valuation\in[0,\rightpoint]$ and location $\loc$ of an
$\rightpoint$-\SPTG, such that
$\val(\loc,\valuation)\notin \{-\infty,+\infty\}$, a memoryless
optimal strategy for $\MaxPl$, and an optimal switching strategy for
$\MinPl$: a switching strategy is described by a pair
$(\stratmin^1,\stratmin^2)$ of memoryless strategies and a switch
threshold $K$, so that the optimal strategy is obtained by playing
$\stratmin^1$ until the value of the current finite play is below $K$,
in which case, we switch to strategy $\stratmin^2$, that can be taken
as an attractor strategy, that only wants to reach a final location.

\subsection{Study of the complete value functions: $\game$ is finitely
  optimal}

Still for an $\rightpoint$-\SPTG with only urgent locations, we now
study a precise characterisation of the functions
$\val(\loc)\colon \valuation\in[0,\rightpoint] \mapsto
\val(\loc,\valuation)$,
for all $\loc$, in particular showing that these are cost functions of
$\CF{\{[0,\rightpoint]\}}$.

We first define the set $\F_{\game}$ of affine functions over
$[0,\rightpoint]$ as follows:
\[\F_{\game} = \{k+\fgoal_\loc\mid \loc\in\LocsFin \land 
k\in[-(|\Locs|-1)\maxPriceTrans,|\Locs|\maxPriceTrans]\cap\Z\}\]
Observe that this set is finite and that its cardinality is
$2|\Locs|^2\maxPriceTrans$, pseudo-polynomial in the size of
$\game$. Moreover, as a direct consequence of
Corollary~\ref{cor:possible-values}, this set contains enough
information to compute the value of the game in each possible
valuation of the clock, in the following sense:
\begin{lemma}\label{lem:for-all-val-there-is-f-in-F}
  For all $\loc\in \Locs$, for all $\valuation\in [0,\rightpoint]$: if
  $\val(\loc,\valuation)$ is finite, then there is $f\in\F_{\game}$
  such that $\val(\loc,\valuation)=f(\valuation)$.
\end{lemma}

We compute the set of intersections of two affine functions of
$\F_\game$:
\[\posscp_\game =\{\valuation\in [0,\rightpoint]\mid \exists
f_1,f_2\in\F_{\game}\quad f_1\neq f_2\land
f_1(\valuation)=f_2(\valuation)\}\,.\]
This set is depicted in Figure~\ref{fig:network-posscp} on an
example. Observe that $\posscp_\game$ contains at most
$|\F_{\game}|^2$ points since all functions from $\F_{\game}$ are
affine, hence they can intersect at most once with every other
function. Thus, the cardinality of $\posscp_\game$ is
$4|\LocsFin|^4(\maxPriceTrans)^2$, also bounded by a pseudo-polynomial
in the size of~$\game$. Moreover, since all functions of $\F_{\game}$
take rational values in $0$ and $\rightpoint\in\Q$, we know that
$\posscp_\game\subseteq \Q$. This set contains all the cutpoints of
the value function of $\game$, as shown in
Proposition~\ref{prop:baseCase}.

\begin{figure}[tbp]
  \centering
  \newdimen\Xabs
  \newdimen\Xord
  \begin{tikzpicture}
    \node[below] at (7,0) {$\valuation$};
    \draw[->,very thin] (-0.3,0) -- (7,0);
    \draw[->,very thin] (0,-1.7) -- (0,5);
    \node[below left] at (0,0) {0};
    \node[below right] at (6,0) {$\rightpoint$};
    \draw[dotted] (6,-1.7) -- (6,4.5);

    \draw[color=purple,draw,name path = fu] (0,3) -- (6,4);
    \path[very thick,color=purple,draw,name path = f] (0,2) -- (6,3);
    \draw[color=purple,draw,name path = fd] (0,1) -- (6,2);

    \draw[color=cyan,draw,name path = gu] (0,4.5) -- (6,0.5);
    \draw[very thick,color=cyan,draw,name path = g] (0,3.5) -- (6,-0.5);
    \draw[color=cyan,draw,name path = gd] (0,2.5) -- (6,-1.5);
    
    \draw[color=olive,draw,name path = hu] (0,2.5) -- (6,2.5);
    \draw[very thick,color=olive,draw,name path = h] (0,1.5) -- (6,1.5);
    \draw[color=olive,draw,name path = hd] (0,0.5) -- (6,0.5);

    \foreach \fA/\fB in {gd/hu, g/fu, g/hu, fu/gu, f/gu, gu/fd, h/gu, gu/hd} {
      \path[name intersections={of={\fA} and {\fB},by={X}}] (X);
      \pgfgetlastxy{\Xabs}{\Xord};
      \node[fill,circle,inner sep=0pt,minimum size=2mm] (Y) at
      (\Xabs,0) {};
      \draw[dotted,thick] (X) -- (Y);
    }

  \end{tikzpicture}
  \caption{Network of affine functions defined by $\F_\game$:
    functions in bold are final affine functions of $\game$, whereas
    non-bold ones are their translations with weights
    $k\in[-(|\Locs|-1)\maxPriceTrans,|\Locs|\maxPriceTrans]\cap\Z$.
    $\posscp_\game$ is the set of abscisses of intersections points,
    represented by black disks.}
  \label{fig:network-posscp}
\end{figure}

Notice, that this result allows us to compute $\val(\loc)$
for every $\loc\in\Locs$. First, we compute the set
$\posscp_{\game}=\{y_1,y_2,\ldots,y_\ell\}$, which can be done in
pseudo-polynomial time in the size of $\game$. Then, for all
$1\leq i\leq \ell$, we can compute the vectors
$\big(\val(\loc,y_i)\big)_{\loc\in\Locs}$ of values in each location
when the clock takes value $y_i$ using
\algorithmcfname~\ref{algo:value-iteration-fixed}. This provides the
value of $\val(\loc)$ in each cutpoint, for all locations $\loc$,
which is sufficient to characterise the whole value function, as it is
continuous and piecewise affine. Observe that all cutpoints, and
values in the cutpoints, in the value function are rational numbers,
so \algorithmcfname~\ref{algo:value-iteration-fixed} is effective.
Thanks to the above discussions, this procedure consists in a
pseudo-polynomial number of calls to a pseudo-polynomial algorithm,
hence, it runs in pseudo-polynomial time. This allows us to conclude
that $\val_\game(\loc)$ is a cost function for all $\loc$. This proves
item (iii) of the definition of finite optimality for \SPTG{s} with
only urgent locations

Let us conclude the proof that \SPTG{s} are finitely optimal by
showing that $\MinPl$ has a fake-optimal NC-strategy, and $\MaxPl$ has
an optimal FP-strategy. Let
$\valuation_1, \valuation_2,\ldots,\valuation_k$ be the sequence of
elements from $\posscp_\game$ in increasing order, and let us assume
$\valuation_0=0$. For all $0\leq i\leq k$ let $f_i^\loc$ be the
function from $\F_\game$ that defines the piece of $\val_\game(\loc)$
in the interval $[\valuation_{i-1},\valuation_i]$ (we have shown above
that such an $f_i^\loc$ always exists). Formally, for all
$0\leq i\leq k$, $f_i^\loc\in\F_\game$ verifies
$\val(\loc,\valuation)=f_i^\loc(\valuation)$, for all
$\valuation\in
[\valuation^\loc_{i-1},\valuation^\loc_i]$. 
Next, for all $1\leq i\leq k$, let $\mu_{i}$ be a value taken in the
middle of $[\valuation_{i-1}, \valuation_i]$, i.e.,
$\mu_i=\frac{\valuation_i+\valuation_{i-1}}{2}$. Note that all
$\mu_i$'s are rational values since all $\valuation_i$'s are. By
applying \SolveInstant in each $\mu_i$, we can compute
$(\val_\game(\loc,\mu_i))_{\loc\in\Locs}$, and we can extract an
optimal memoryless strategy $\stratmax^i$ for $\MaxPl$ and an optimal
switching strategy $\stratmin^i$ for $\MinPl$. Thus we know that, for
all $\loc\in\Locs$, playing $\stratmin^i$ (respectively,
$\stratmax^i$) from $(\loc, \mu_i)$ allows $\MinPl$ (respectively,
$\MaxPl$) to ensure a cost at most (respectively, at least)
$\val_\game(\loc,\mu_i)=f_i^\loc(\mu_i)$. However, it is easy to check
that the bound given by $f_i^\loc(\mu_i)$ holds in every valuation,
i.e., for all $\loc$, for all $\valuation$
\[\valgs{\game}{\stratmin^i}(\loc,\valuation)\leq f_i^\loc(\valuation)
\qquad \text{ and } \qquad
\valgs{\game}{\stratmax^i}(\loc,\valuation)\geq
f_i^\loc(\valuation)\,.\] This holds because:
\begin{inparaenum}[$(i)$]
\item $\MinPl$ can play $\stratmin^i$ from all clock valuations (in
  $[0,r]$) since we are considering an $r$-\SPTG; and
\item $\MaxPl$ does not have more possible strategies from an
  arbitrary valuation $\valuation\in [0,r]$ than from~$\mu_i$, because
  all locations are urgent and time can not elapse (neither from
  $\valuation$, nor from $\mu_i$).
\end{inparaenum}
And symmetrically for $\MaxPl$.

We conclude that $\MinPl$ can
consistently play the same strategy $\stratmin^i$ from all
configurations $(\loc,\valuation)$ with
$\valuation\in [\valuation_{i-1},\valuation_i]$ and secure a
cost which is at most
$f_i^\loc(\valuation)=\val_\game(\loc,\valuation)$, i.e.,
$\stratmin^i$ is optimal on this interval. By definition of
$\stratmin^i$, it is easy to extract from it a fake-optimal
NC-strategy (actually, $\stratmin^i$ is a switching strategy described
by a pair $(\stratmin^1,\stratmin^2)$, and $\stratmin^1$ can be used
to obtain the fake-optimal NC-strategy). The same reasoning applies to
strategies of $\MaxPl$ and we conclude that $\MaxPl$ has an optimal
FP-strategy.

\section{Every \SPTG is finitely optimal}


We start with an auxiliary lemma showing a property of the rates of
change of the value functions associated to non-urgent locations

\begin{lemma}\label{lem:rate}\todo{%
    Pourquoi ne pas définir la fonction $\textit{slope}$, par exemple
    $\textit{slope}^\ell_{\game}(\valuation,\valuation') = \frac{\Value_\game(\loc,\valuation') -
      \Value_\game(\loc,\valuation)}{\valuation'-\valuation}$? A.}%
  Let $\game$ be an $r$-\SPTG, $\loc$ and $\loc'$ be non-urgent
  locations of $\MinPl$ and $\MaxPl$, respectively. Then for all
  $0\leq \valuation<\valuation'\leq r$:
  \[\frac{\Value_\game(\loc,\valuation') -
    \Value_\game(\loc,\valuation)}{\valuation'-\valuation} \geq
  -\price(\loc)\qquad \textrm{ and } \qquad
  \frac{\Value_\game(\loc',\valuation') -
    \Value_\game(\loc',\valuation)}{\valuation'-\valuation} \leq
  -\price(\loc')\,.\]
\end{lemma}

\begin{proof}
  For the location $\loc$, the inequality rewrites in
  \[\Value_\game(\loc,\valuation) \leq
  (\valuation'-\valuation)\price(\loc) +
  \Value_\game(\loc,\valuation')\,.\]
  Using the upper definition of the value (thanks to the determinacy
  result of Theorem~\ref{thm:determined}), it suffices to prove, for
  all $\varepsilon>0$, the existence of a strategy $\minstrategy$
  such that for all strategies $\maxstrategy$ of the opponent
  \[\cost{\Play{(\loc,\valuation),\minstrategy,\maxstrategy}} \leq
  (\valuation'-\valuation)\price(\loc) +
  \Value_\game(\loc,\valuation') + \varepsilon\,.\]
  The definition of the value implies the existence of a strategy
  $\minstrategy'$ such that for all strategies $\maxstrategy$
  \[\cost{\Play{(\loc,\valuation'),\minstrategy',\maxstrategy}} \leq
  \Value_\game(\loc,\valuation') + \varepsilon\,.\]
  Then, $\minstrategy$ can be obtained by playing from
  $(\loc,\valuation)$, at the first turn, as prescribed by
  $\minstrategy'$ but delaying $\valuation'-\valuation$ time units
  more (that we are allowed to do since $\loc$ is non-urgent), and,
  for other turns, directly like $\minstrategy'$. A similar reasoning
  allows us to obtain the result for $\loc'$.
\end{proof}

Then, we observe that the construction of $\game_r$
does not alter the value of the game:

\begin{lemma}\label{lem:waiting} 
  For all $\valuation\in [0,r]$ and locations $\loc$,
  $\Value_{\game}(\loc,\valuation) =
  \Value_{\game_r}(\loc,\valuation)$.
\end{lemma}

Now, we turn our attention to the construction of
$\game_{\Locs',r}$. We show that, even if the locations in $\Locs'$
are turned into urgent locations, we may still obtain for them a
similar result of the rates of change as the one of
Lemma~\ref{lem:rate}:

\begin{lemma}\label{lem:bounded}
  For all locations $\loc \in \Locs'\cap\LocsMin$ (respectively,
  $\loc \in \Locs'\cap\LocsMax$), and $\valuation \in [0,r]$,
  $\Value_{\game_{\Locs',r}}(\loc,\valuation) \leq
  (r-\valuation)\price(\loc)+\Value_{\game}(\loc,r)$ 
  (respectively,  $\Value_{\game_{\Locs',r}}(\loc,\valuation)
  \geq (r-\valuation)\price(\loc) + \Value_{\game}(\loc,r)\;)$.
\end{lemma}
\begin{proof}
  It suffices to notice that from $(\loc,\valuation)$, $\MinPl$
  (respectively, $\MaxPl$) may choose to go directly in $\loc^f$
  ensuring the value
  $(r-\valuation)\price(\loc)+ \Value_{\game}(\loc,r)$.
\end{proof}

\subsection{Proof of Proposition~\ref{lem:SameValue}}
\label{app:SameValue}

Let $\stratmin$ and $\stratmax$ be a fake-optimal NC-strategy of
$\MinPl$ and an optimal FP-strategy of $\MaxPl$ in
$\game_{\Locs'\cup\{\loc\},r}$, respectively. Notice that both
strategies are also well-defined finite positional strategies in
$\game_{\Locs',r}$.

First, let us show that $\stratmin$ is indeed an NC-strategy in
$\game_{\Locs',r}$. Take a finite play
$(\loc_0,\valuation_0) \xrightarrow{c_0} \cdots \xrightarrow{c_{k-1}}
(\loc_k,\valuation_k)$,
of length $k\geq 2$, that conforms with $\stratmin$ in
$\game_{\Locs',r}$, and with $\loc_0=\loc_k$ and
$\valuation_0,\valuation_k$ in the same interval $I$ of
$\intervals(\stratmin)$. For every $\loc_i$ that is in $\LocsMin$, and
$\valuation\in I$, $\stratmin(\loc_i,\valuation)$ must have a $0$
delay, otherwise $\valuation_k$ would not be in the same interval as
$\valuation_0$. Thus, the play
$(\loc_0,\valuation_0) \xrightarrow{c'_0} \cdots
\xrightarrow{c'_{k-1}} (\loc_k,\valuation_0)$
also conforms with $\stratmin$ (with possibly different costs).
Furthermore, as all the delays are $0$ we are sure that this play is
also a valid play in $\game_{\Locs'\cup\{\loc\},r}$, in which
$\stratmin$ is an NC-strategy. Therefore,
$\price(\loc_0,\loc_1) + \cdots + \price (\loc_{k-1},\loc_k)\leq -1$,
and $\stratmin$ is an NC-strategy in $\game_{\Locs',r}$.

We now show the result for $\loc\in \LocsMin$. The proof for
$\loc\in\LocsMax$ is a straightforward adaptation. Notice that every
play in $\game_{\Locs',r}$ that conforms with $\stratmin$ is also a
play in $\game_{\Locs'\cup\{\loc\},r}$ that conforms with $\stratmin$,
as $\stratmin$ is defined in $\game_{\Locs'\cup\{\loc\},r}$ and thus
plays with no delay in location $\loc$. Thus, for all
$\valuation\in [a,r]$ and $\loc'\in \locs$, by the optimality result
of Lemma~\ref{lem:fake-optimality},
\begin{equation}\label{eq:fake-value-inequ}
  \Value_{\game_{\Locs',r}}(\loc',\valuation) \leq
  \fakeValue_{\game_{\Locs',r}}^{\stratmin}(\loc',\valuation) =
  \fakeValue_{\game_{\Locs'\cup\{\loc\},r}}^{\stratmin}(\loc',\valuation)=
  \Value_{\game_{\Locs'\cup\{\loc\},r}}(\loc',\valuation)\,.
\end{equation}

To obtain that
$\Value_{\game_{\Locs',r}}(\loc',\valuation)=
\Value_{\game_{\Locs'\cup\{\loc\},r}}(\loc',\valuation)$,
it remains to show the reverse inequality. To that extent, let $\run$
be a finite play in $\game_{\Locs',r}$ that conforms with $\stratmax$,
starts in a configuration $(\loc',\valuation)$ with
$\valuation\in [a,r]$, and ends in a final location. We show by
induction on the length of $\run$ that
$\cost{\run}\geq
\Value_{\game_{\Locs'\cup\{\loc\},r}}(\loc',\valuation)$.
If $\run$ has size $1$ then $\loc'$ is a final configuration and
$\cost{\run} = \Value_{\game_{\Locs'\cup\{\loc\},r}}(\loc',\valuation)
= \fgoal'_{\loc'}(\valuation)$.

Otherwise $\run= (\loc',\valuation) \xrightarrow{c} \run'$ where
$\run'$ is a run that conforms with $\stratmax$, starting in a
configuration $(\loc'',\valuation'')$ and ending in a final
configuration. By induction hypothesis, we have
$\cost{\run'} \geq
\Value_{\game_{\Locs'\cup\{\loc\},r}}(\loc'',\valuation'')$.
We now distinguish three cases, the two first being immediate:
\begin{itemize}
\item \underline{If $\loc'\in \LocsMax$}, then
  $\stratmax (\loc',\valuation)$ leads to the next configuration
  $(\loc'',\valuation'')$, thus
  \[\Value_{\game_{\Locs'\cup\{\loc\},r}}(\loc',\valuation) =
  \Value_{\game_{\Locs'\cup\{\loc\},r}}^{\stratmax}(\loc',\valuation)= c +
  \Value_{\game_{\Locs'\cup\{\loc\},r}}^{\stratmax}(\loc'',\valuation'') \leq c+
  \cost{\run'} = \cost{\run}\,.
  \]

\item \underline{If $\loc'\in \LocsMin$, and $\loc'\neq \loc$ or
    $\valuation''=\valuation$}, we have that
  $(\loc',\valuation) \xrightarrow{c} (\loc'',\valuation'')$ is a
  valid transition in $\game'$. Therefore,
  $\Value_{\game_{\Locs'\cup\{\loc\},r}}(\loc',\valuation) \leq c +
  \Value_{\game_{\Locs'\cup\{\loc\},r}}(\loc'',\valuation'')$, hence
  \[\cost{\run} = c+\cost{\run'} \geq c+
  \Value_{\game_{\Locs'\cup\{\loc\},r}}(\loc'',\valuation'') \geq
  \Value_{\game_{\Locs'\cup\{\loc\},r}}(\loc',\valuation).\]

\item Finally, \underline{if $\loc'=\loc$ and
    $\valuation'' >\valuation$}, then
  $c = (\valuation''-\valuation)\price(\loc)+\price (\loc,\loc'')$.
  As
  $(\loc,\valuation'') \xrightarrow{\price (\loc,\loc'')}
  (\loc'',\valuation'')$
  is a valid transition in $\game_{\Locs'\cup\{\loc\},r}$, we have
  $\Value_{\game_{\Locs'\cup\{\loc\},r}}(\loc,\valuation'') \leq
  \price (\loc,\loc'') +
  \Value_{\game_{\Locs'\cup\{\loc\},r}}(\loc'',\valuation'')$.
  Furthermore, since $\valuation''\in [a,r]$, we can use
  \eqref{eq:SlopeBigger} to obtain
  \[
  \Value_{\game_{\Locs'\cup\{\loc\},r}}(\loc,\valuation) \leq
  \Value_{\game_{\Locs'\cup\{\loc\},r}}(\loc,\valuation'') +
  (\valuation''-\valuation)\price(\loc)\leq
  \Value_{\game_{\Locs'\cup\{\loc\},r}}(\loc'',\valuation'') +
  \price(\loc,\loc'') + (\valuation''-\valuation)\price(\loc)\,.
  \]
  Therefore
  \begin{align*}
    \cost{\run} 
    &=
      (\valuation''-\valuation) \price(\loc)+
      \price(\loc,\loc'')+\cost{\run'} \\
    &\geq  
      (\valuation''-\valuation) \price(\loc)+ \price(\loc,\loc'')+
      \Value_{\game_{\Locs'\cup\{\loc\},r}}(\loc'',\valuation'') 
      \geq \Value_{\game_{\Locs'\cup\{\loc\},r}}(\loc',\valuation)\,.
  \end{align*}
\end{itemize}
This concludes the induction.  As a consequence,
\[\inf_{\stratmin'\in \stratsofmin(\game_{\Locs',r})}
\costgame{\game_{\Locs',r}}{\Play{(\loc',\valuation),\stratmin',\stratmax}}
\geq \Value_{\game_{\Locs'\cup\{\loc\},r}}(\loc',\valuation)\]
for all locations $\loc'$ and $\valuation\in [a,r]$, which finally
proves that
$\Value_{\game_{\Locs',r}}(\loc',\valuation)\geq
\Value_{\game_{\Locs'\cup\{\loc\},r}}(\loc',\valuation)$.
Fake-optimality of $\minstrategy$ over $[a,r]$ in
$\game_{\Locs'\cup\{\loc\},r}$ is then obtained by
\eqref{eq:fake-value-inequ}.

\subsection{Proof that $\Next(r)<r$}
\label{app:operatorNext}
This lemma allows us to effectively compute $\Next(r)$:

\begin{lemma}\label{lem:operatorNext}
  Let $\game$ be an $\SPTG$, $\Locs'\subseteq \Locs\setminus\LocsUrg$,
  and $r\in(0,1]$, such that $\game_{\Locs'',r}$ is finitely optimal
  for all $\Locs''\subseteq \Locs'$.
  Then, $\Next_{\Locs'}(r)$ is the minimal valuation such that for all
  locations $\loc\in \Locs'\cap\LocsMin$ (respectively,
  $\loc \in \Locs'\cap\LocsMax$), the slopes of the affine sections of
  the cost function $\Value_{\game_{\Locs',r}}(\loc)$ on
  $[\Next(r),r]$ are at least (respectively, at most) $-\price(\loc)$.
  Moreover, $\Next(r)<r$.
\end{lemma}

\begin{proof}
  Since $\Value_{\game_{\Locs',r}}(\loc)=\Value_\game(\loc)$ on
  $[\Next(r),r]$, and as $\loc$ is non-urgent in $\game$,
  Lemma~\ref{lem:rate} states that all the slopes of
  $\Value_\game(\loc)$ are at least (respectively, at most)
  $-\price(\loc)$ on $[\Next(r),r]$.

  We now show the minimality property by contradiction. Therefore, let
  $r'<\Next(r)$ such that all cost functions
  $\Value_{\game_{\Locs',r}}(\loc)$ are affine on $[r',\Next(r)]$, and
  assume that for all $\loc\in \Locs'\cap\LocsMin$ (respectively,
  $\loc \in \Locs'\cap\LocsMax$), the slopes of
  $\Value_{\game_{\Locs',r}}(\loc)$ on $[r',\Next(r)]$ are at least
  (respectively, at most) $-\price(\loc)$. Hence, this property holds
  on $[r',r]$. Then, by applying Proposition~\ref{lem:SameValue}
  $|\Locs'|$ times (here, we use the finite optimality of the games
  $\game_{\Locs'',r}$ with $\Locs''\subseteq \Locs'$), one can show
  that for all $\valuation\in [r',r]$
  $\Value_{\game_{r}}(\loc,\valuation) =
  \Value_{\game_{\Locs',r}}(\loc,\valuation)$.
  Using Lemma~\ref{lem:waiting}, we also know that for all
  $\valuation\leq r$, and $\loc$,
  $\Value_{\game_{r}}(\loc,\valuation)=\Value_\game(\loc,\valuation)$.
  Thus,
  $\Value_{\game_{r,\Locs'}}(\loc,\valuation) =
  \Value_\game(\loc,\valuation)$.
  As $r'<\Next(r)$, this contradicts the definition of
  $\Next_{\Locs'}(r)$.

  We finally prove that $\Next(r)<r$. This is immediate in case
  $\Next(r)=0$, since $r>0$. Otherwise, from the result obtained
  previously, we know that there exists $r'<\Next(r)$, and
  $\locMin\in \Locs'$ such that $\Value_{\game_{\Locs',r}}(\locMin)$
  is affine on $[r',\Next(r)]$ of slope smaller (respectively,
  greater) than $-\price(\locMin)$ if $\locMin\in\LocsMin$
  (respectively, $\locMin\in\LocsMax$), i.e.,
  \[
  \begin{cases}
    \Value_{\game_{\Locs',r}}(\locMin,r') >
    \Value_{\game_{\Locs',r}}(\locMin,\Next(r))+(\Next(r)-r')
    \price(\locMin) & \text{if } \locMin\in \LocsMin \\
    \Value_{\game_{\Locs',r}}(\locMin,r') <
    \Value_{\game_{\Locs',r}}(\locMin,\Next(r))+(\Next(r)-r')
    \price(\locMin) & \text{if } \locMin\in \LocsMax \,.
  \end{cases}
  \]
  From Lemma~\ref{lem:bounded}, we also know that
  \[
  \begin{cases}
    \Value_{\game_{\Locs',r}}(\locMin,r') \leq
    \Value_{\game_{\Locs',r}}(\locMin,r)+(r-r') \price(\locMin) &
    \text{ if } \locMin\in \LocsMin \\
    \Value_{\game_{\Locs',r}}(\locMin,r') \geq
    \Value_{\game_{\Locs',r}}(\locMin,r)+(r-r') \price(\locMin) &
    \text{ if } \locMin\in \LocsMax\,.
  \end{cases}\] Both equations combined imply
  \[
  \begin{cases}
    \Value_{\game_{\Locs',r}}(\locMin,r) >
    \Value_{\game_{\Locs',r}}(\locMin,\Next(r)) + (\Next(r)-r)
    \price(\locMin) & \text{if } \locMin\in \LocsMin \\
    \Value_{\game_{\Locs',r}}(\locMin,r) <
    \Value_{\game_{\Locs',r}}(\locMin,\Next(r)) + (\Next(r)-r)
    \price(\locMin) & \text{if } \locMin\in \LocsMax
  \end{cases}
  \]
  which is not possible if $\Next(r)=r$.
\end{proof}

\subsection{Pieces of the value functions are segments of $\F_\game$}
\label{app:eqGoal}

\begin{lemma}\label{lem:eqGoal}
  Assume that $\game_{\locMin,r}$ is finitely optimal. If
  $\Value_{\game_{\locMin,r}}(\locMin)$ is affine on a non-singleton
  interval $I\subseteq [0,r]$ with a slope greater\footnote{For this
    result, the order does not depend on the owner of the location,
    but rather depends on the fact that $\locMin$ has minimal price
    amongst locations of $\game$.}  than $-\price(\locMin)$, then
  there exists $f\in \F_\game$ such that for all $\valuation \in I$,
  $\Value_{\game_{\locMin,r}}(\locMin,\valuation) = f(\valuation)$.
\end{lemma}
\begin{proof}
  Let $\stratmin$ and $\stratmax$ be some fake-optimal NC-strategy and
  optimal FP-strategy in $\game_{\locMin,r}$. As $I$ is a
  non-singleton interval, there exists a subinterval $I'\subset I$,
  which is not a singleton and is contained in a interval of
  $\stratmin$ and of $\stratmax$.

  Let $\valuation \in I'$. As already noticed in the proof of
  Lemma~\ref{lem:r_2-r_1-r_0}, the play
  $\Play{(\locMin,\valuation),\stratmin,\stratmax}$ necessarily
  reaches a final location and has cost
  $\Value_{\game_{\locMin,r}}(\locMin,\valuation)$. Let
  $(\loc_0,\valuation_0) \xrightarrow{c_0} \cdots
  (\loc_k,\valuation_k)$
  be its prefix until the first final location $\loc_k$ (the prefix
  used to compute the cost of the play).  We also let
  $\valuation'\in I'$ be a valuation such that
  $\valuation<\valuation'$.

  Assume by contradiction that there exists an index $i$ such that
  $\valuation<\valuation_i$ and let $i$ be the smallest of such
  indices. For each $j< i$, if $\loc_j\in \LocsMin$, let
  $(t,\transition) = \stratmin(\loc_j,\valuation)$ and
  $(t',\transition') = \stratmin(\loc_j,\valuation')$. Similarly, if
  $\loc_j\in \LocsMax$, we let
  $(t,\transition) = \stratmax(\loc_j,\valuation)$ and
  $(t',\transition') = \stratmax(\loc_j,\valuation')$. As $I'$ is
  contained in an interval of $\stratmin$ and $\stratmax$, we have
  $\transition = \transition'$ and either $t = t' = 0$, or
  $\valuation + t= \valuation'+t'$. Applying this result for all
  $j<i$, we obtain that
  $(\loc_0,\valuation') \xrightarrow{c'_0} \cdots
  (\loc_{i-1},\valuation') \xrightarrow{c'_{i-1}}
  (\loc_i,\valuation_i) \xrightarrow{c_i} \cdots
  (\loc_k,\valuation_k)$
  is a prefix of $\Play{(\locMin,\valuation'),\stratmin,\stratmax}$:
  notice moreover that, as before, this prefix has cost
  $\Value_{\game_{\locMin,r}}(\locMin,\valuation')$.  In particular,
  \[\Value_{\game_{\locMin,r}}(\locMin,\valuation') =
  \Value_{\game_{\locMin,r}}(\locMin,\valuation)
  -(\valuation'-\valuation) \price(\loc_{i-1})\leq
  \Value_{\game_{\locMin,r}}(\locMin,\valuation)
  -(\valuation'-\valuation) \price(\locMin)\]
  which implies that the slope of
  $\Value_{\game_{\locMin,r}}(\locMin)$ is at most $-\price(\locMin)$,
  and therefore contradicts the hypothesis. As a consequence, we have
  that $\valuation_i = \valuation$ for all $i$.
  
  Again by contradiction, assume now that $\loc_k = \loc^f$ for some
  $\loc\in \Locs\setminus\LocsUrg$. By the same reasoning as before,
  we then would have
  $\Value_{\game_{\locMin,r}}(\locMin,\valuation') =
  \Value_{\game_{\locMin,r}}(\locMin,\valuation)
  -(\valuation'-\valuation) \price(\loc)$,
  which again contradicts the hypothesis.

  Therefore, $\loc_k\in \LocsFin$. If we let
  $w=\price(\loc_0,\loc_1)+\cdots+\price(\loc_{k-1},\loc_k)$, we have
  $\Value_{\game_{\locMin,r}}(\locMin,\valuation) =
  w+\fgoal_{\loc_k}(\valuation)$.
  Since $\stratmin$ and $\stratmax$ are FP-strategies, that play
  constantly in valuation $\valuation$, we know that
  $(\loc_0,\valuation) \xrightarrow{c_0} \cdots (\loc_k,\valuation)$
  has no cycle, therefore
  $w\in [-(|\Locs|-1)\maxPriceTrans,|\Locs|\maxPriceTrans]\cap \Z$.
  Notice that the previous developments also show that for all
  $\valuation'\in I'$ (here, $\valuation<\valuation'$ is not needed),
  $\Value_{\game_{\locMin,r}}(\locMin,\valuation') =
  w+\fgoal_{\loc_k}(\valuation')$,
  with the same location $\loc_k$, and weight $k$. Since this equality
  holds on $I'\subseteq I$ which is not a singleton, and
  $\Value_{\game_{\locMin,r}}(\locMin)$ is affine on $I$, it holds
  everywhere on $I$.
\end{proof}

\subsection{Proof of Lemma~\ref{lem:stationarysequence-locMin}}
\label{app:stationarysequence-locMin}

For the first item, we assume $\locMin\in \LocsMin$, since the proof
of the other case only differ with respect to the sense of the
inequalities. From Lemma~\ref{lem:operatorNext}, we know that in
$\game_{\locMin,r_i}$ there exists $r'<r_{i+1}$ such that
$\Value_{\game_{\locMin,r_i}}(\locMin)$ is affine of $[r',r_{i+1}]$
and its slope is smaller that $-\price(\locMin)$, i.e.,
$\Value_{\game_{\locMin,r_i}}(\locMin,r_{i+1})<
\Value_{\game_{\locMin,r_i}}(r') - (r_{i+1}-r')\price(\locMin)$.
Lemma~\ref{lem:bounded} also ensures that
$\Value_{\game_{\locMin,r_i}}(\locMin,r')\leq
\Value_\game(\locMin,r_i) + (r_i-r')\price(\locMin)$.
Combining both inequalities allows us to conclude.

\medskip
We now turn to the proof of the second item, showing the stationarity
of sequence $(r_i)$. We consider first the case where
\underline{$\locMin\in \LocsMax$}.  Let $i>0$ such that $r_i \neq 0$
(if there exist no such $i$ then $r_1=0$). Recall from
Lemma~\ref{lem:operatorNext} that there exists $r'_i<r_i$ such that
$\Value_{\game_{\locMin,r_{i-1}}}(\locMin)$ is affine on $[r'_i,r_i]$,
of slope greater than $-\price(\locMin)$. In particular,
\[
\frac{\Value_{\game_{\locMin,r_{i-1}}}(\locMin,r_i)-
  \Value_{\game_{\locMin,r_{i-1}}}(\locMin,r'_i)} {r_i-r'_i} >
-\price(\locMin)\,.\]
Lemma~\ref{lem:eqGoal} states that on $[r'_i,r_i]$,
$\Value_{\game_{\locMin,r_{i-1}}}(\locMin)$ is equal to some
$f_i\in \F_\game$. As $f_i$ is an affine function,
$f_i(r_i) = \Value_{\game_{\locMin,r_{i-1}}}(\locMin,r_i)$, and
$f_i(r'_i) = \Value_{\game_{\locMin,r_{i-1}}}(\locMin,r'_i)$, for
all~$\valuation$,
\[f_i(\valuation) = \Value_{\game_{\locMin,r_{i-1}}}(\locMin,r_i) +
\frac{\Value_{\game_{\locMin,r_{i-1}}}(\locMin,r'_i)-
  \Value_{\game_{\locMin,r_{i-1}}}(\locMin,r_i)} {r_i - r'_i} (r_i
-\valuation).\]
Since $\game_{\locMin,r_{i-1}}$ is assumed to be finitely optimal, we
know that
$\Value_{\game_{\locMin,r_{i-1}}}(\locMin,r_i) =
\Value_{\game}(\locMin,r_i)$,
by definition of $r_i=\Next_{\locMin}(r_{i-1})$.  Therefore, for all
valuation $\valuation<r_i$, we have
$f_i(\valuation) < \Value_\game(\locMin,r_i) +\price(\locMin)
(r_i-\valuation)$.

Consider then $j> i$ such that $r_j \neq 0$.  We claim that
$f_j \neq f_i$.  Indeed, we have $\Val_\game(\locMin,r_j) = f_j(r_j)$.
As, in $\game$, $\locMin$ is a non-urgent location,
Lemma~\ref{lem:rate} ensures that ($\star$):
$\Value_\game(\locMin,r_j) \geq \Value_\game(\locMin,r_i)
+\price(\locMin) (r_i-r_j)$.
As for all $i'$, $\Value_\game(\locMin,r_{i'}) = f_{i'}(r_{i'})$,
($\star$) is equivalent to
$f_j(r_j) \geq f_i(r_i) +\price(\locMin) (r_i-r_j)$.  Recall that
$f_i$ has a slope strictly greater that $-\price(\locMin)$, therefore
$f_i(r_j) < f_i(r_i) +\price(\locMin) (r_i-r_j) \leq f_j(r_j)$. As a
consequence $f_i \neq f_j$ (this is depicted in
Figure~\ref{fig:fineqfj}).

\begin{figure}\centering
  \begin{tikzpicture}
    \node[above] at (0,5) {$\Value_{\game}(\locMin,\valuation)$};
    \node[below] at (7,0) {$\valuation$};
    \draw[->] (-0.3,0) -- (7,0);
    \draw[->] (0,-0.3) -- (0,5);
    \draw[gray] (6,0) -- (6,5);
    \draw[gray] (3,0) -- (3,5);
    \draw[dashed] (7,4) -- (-0.5,0.25);
    \node[above] at (7,4) {$-\price(\locMin)$};
    
    \draw (6,3.5) -- ++(-0.75,-0.375) -- ++ (-0.75,0.2) -- ++ (-0.75,-0.1) -- ++ (-0.75,-0.375) -- ++ (-1,-0.5) -- ++ (-1,0.4);
    \draw (6,3.5) -- ++(0.5,-0.3);
    \draw[dashed] (3,2.85) -- ++(-2,-1);
    
    \draw[dotted] (6,3.5) -- ++(-1.5,-2);
    \draw[dotted] (3,2.85) -- ++(-1.5,-1.5);
    \node[below] at (4.5,1.5) {$f_i$};
    \node[above] at (1.5,1.35) {$f_j$};
    
    \node[below] at (3,0) {$r_j$};
    \node[below] at (6,0) {$r_i$};

  \end{tikzpicture}
  
  \caption{The case $\locMin\in\LocsMax$: a geometric proof of
    $f_i\neq f_j$. The dotted lines represents $f_i$ and $f_j$, the
    dashed lines have slope $-\price(\locMin)$, and the plain line
    depicts $\Value_\game(\locMin,\cdot)$. Because the slope of $f_i$
    is strictly smaller than $-\price(\locMin)$, and the value at
    $r_j$ is above the dashed line it can not be the case that
    $f_i(r_j) = \Value_\game(\locMin,r_j)=f_j(r_j)$.}
  \label{fig:fineqfj}
\end{figure}

Therefore, there can not be more than $|\F_\game|+1$ non-null elements
in the sequence $r_0\geq r_1\geq \cdots$, which proves that there
exists $i\leq |\F_\game|+2$ such that $r_i=0$.

\medskip We continue with the case where
\underline{$\locMin\in \LocsMin$}. Let
$r_\infty = \inf\{r_i \mid i\geq 0\}$.  In this case, we look at the
affine parts of $\Val_\game(\locMin)$ with a slope greater than
$-\price(\locMin)$, and we show that there can only be finitely many
such segment in $[r_\infty,1]$.  We then show that there is at least
one such segment contained in $[r_{i+1},r_i]$ for all $i$, bounding
the size of the sequence.

In the following, we call \emph{segment} every interval
$[a,b]\subset (r_\infty,1]$ such that $a$ and $b$, are two
consecutive cutpoints of the cost function $\Val_\game(\locMin)$
over $(r_\infty,1]$. Recall that it means that $\Val_\game(\locMin)$
is affine on $[a,b]$, and if we let $a'$ be the greatest cutpoint
smaller than $a$, and $b'$ the smallest cutpoint greater than $b$,
the slopes of $\Val_\game(\locMin)$ on $[a',a]$ and $[b,b']$ are
different from the slope on $[a,b]$.  We abuse the notations by
referring to \emph{the slope of a segment $[a,b]$} for the slope of
$\Val_\game(\locMin)$ on $[a,b]$ and simply call \emph{cutpoint} a
cutpoint of $\Val_\game(\locMin)$.

To every segment $[a,b]$ with a slope greater than
$-\price(\locMin)$, we associate a function $f_{[a,b]}\in \F_\game$
as follows. Let $i$ be the smallest index such that
$[a,b]\cap[r_{i+1},r_i]$ is a non singleton interval $[a',b']$.
Lemma~\ref{lem:eqGoal} ensures that there exists
$f_{[a,b]}\in \F_\game$ such that for all $\valuation \in [a',b']$,
$\Value_\game(\locMin,\valuation) = f_{[a,b]}(\valuation)$.

Consider now two disjoint segments $[a,b]$ and $[c,d]$ with a slope
strictly greater than $-\price(\locMin)$, and assume that
$f_{[a,b]}=f_{[c,d]}$ (in particular both segments have the same
slope). Without loss of generality, assume that $b<c$. We claim that
there exists a segment $[e,g]$ in-between $[a,b]$ and $[c,d]$ with a
slope greater than the slope of $[c,d]$, and that $f_{[e,g]}$ and
$f_{[a,b]}$ intersect over $[b,c]$, in a point of abscisse $x$, i.e.,
$x\in [b,c]$ verifies $f_{[e,g]}(x) = f_{[a,b]}(x)$ (depicted in
Figure~\ref{fig:pureGeometry}).

\begin{figure}\centering
  \begin{tikzpicture}
    \draw (0,0) -- (1,1) -- (2,0.5) -- (2.25,4) -- (4,6.5) -- (5,5) -- (6,6);
    \draw[dashed] (-0.5,-0.5) -- (6.5,6.5);

    \node[above] at (0,0) {$a$}; 
    \node[above] at (1,1) {$b$};
    \node[below] at (2,0.5) {$g$}; 
    \node[above] at (2.25,4) {$e$};
    \node[above] at (4,6.5) {$\alpha$}; 
    \node[below] at (5,5) {$c$}; 
    \node[below] at (6,6) {$d$};
    \node at (2.115,2.115) {$\bullet$};
    \node[below right] at (2.115,2.115) {$x$};
  \end{tikzpicture}
  \caption{In order for the segments $[a,b]$ and $[c,d]$ to be aligned, there must exist a segment with a biggest slope crossing $f_[a,b]$ (represented by a dashed line) between $b$ and $c$.}
  \label{fig:pureGeometry}
\end{figure}
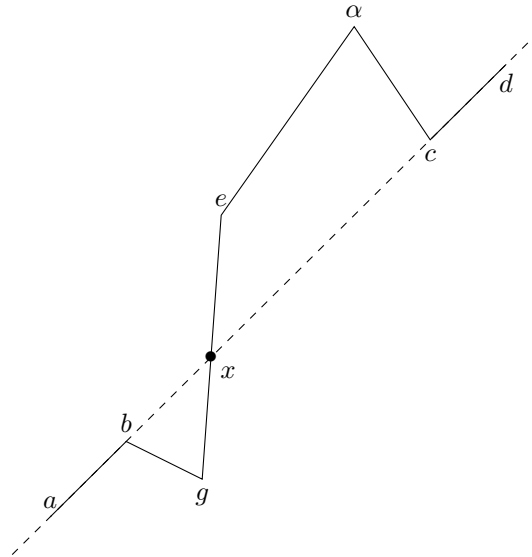

Let $\alpha$ be the greatest cutpoint smaller than $c$. We know that
the slope of $[\alpha,c]$ is different from the one of $[c,d]$.  If it
is greater then define $e=\alpha$ and $x=g=c$, those indeed satisfy
the property.  If the slope of $[\alpha,c]$ is smaller than the one of
$[c,d]$, then for all $\valuation \in [\alpha,c)$,
$\Val_\game(\locMin,\valuation) > f_{[c,d]}(\valuation)$.  Let $x$ be
the greatest point in $[b,\alpha]$ such that
$\Val_\game(\locMin,x) = f_{[c,d]}(x)$. We know that it exists since
$\Val_\game(\locMin,b) = f_{[c,d]}(b)$, and $\Val_\game(\locMin)$ is
continuous.  Observe that
$\Val_\game(\locMin,\valuation) > f_{[c,d]}(\valuation)$, for all
$x<\valuation<c$.  Finally, let $g$ be the smallest cutpoint of
$\Val_\game(\locMin)$ strictly greater than $x$, and $e$ the greatest
cutpoint of $\Val_\game(\locMin)$ smaller than or equal to $x$. By
construction $[e,g]$ is a segment that contains $x$. The slope of the
segment $[e,g]$ is
$s_{[e,g]}=\frac{\Val_\game(\locMin,g) - \Val_\game(\locMin,x)}{g-x}$,
and the slope of the segment $[c,d]$ is equal to
$s_{[c,d]} = \frac{f_{[c,d]}(g) - f_{[c,d]}(x)}{g-x}$. Remembering
that $\Val_\game(\locMin,x) = f_{[c,d]}(x)$, and that
$\Val_\game(\locMin,g) > f_{[c,d]}(g)$ since $g\in (x,c)$, we obtain
that $s_{[e,g]} > s_{[c,d]}$. Finally, since
$\Val_\game(\locMin,x) = f_{[c,d]}(x) = f_{[e,g]}(x)$, it is indeed
the abscisse of the intersection point of $f_{[c,d]}=f_{[a,b]}$ and
$f_{[e,g]}$, which concludes the proof of the previous claim.

For every function $f\in \F_\game$, there are less than $|\F_\game|$
intersection points between $f$ and the other functions of $\F_\game$
(at most one for each pair $(f,f')$). If $f$ has a slope greater than
$-\price(\locMin)$, thanks to the previous paragraph, we know that
there are at most $|\F_\game|$ segments $[a,b]$ such that
$f_{[a,b]}=f$. Summing over all possible functions $f$, there are at
most $|\F_\game|^2$ segments with a slope greater than
$-\price(\locMin)$.

Now, we link those segments with the valuations $r_i$'s, for $i>0$.
By item $(i)$, thanks to the finite-optimality of
$\game_{\locMin,r_{i}}$,
$\Val_\game(\locMin,r_{i+1}) < (r_i-r_{i+1}) \price(\locMin) +
\Val_\game(\locMin,r_i)$.
Furthermore, Lemma~\ref{lem:r_2-r_1-r_0} states that the slope of the
segment directly on the left of $r_i$ is equal to
$-\price(\locMin)$. With the previous inequality in mind, this can not
be the case if $\Value_\game(\locMin)$ is affine over the whole
interval $[r_{i+1},r_{i}]$. Thus, there exists a segment $[a,b]$ of
slope strictly greater than $-\price(\locMin)$ such that
$b\in [r_{i+1},r_{i}]$. As we also know that the slope left to
$r_{i+1}$ is $-\price(\locMin)$, it must be the case that
$a\in [r_{i+1},r_i]$. Hence, we have shown that in-between $r_{i+1}$
and $r_i$, there is always a segment (this is depicted in
Figure~\ref{fig:slopesLocMin}). As the number of such segments is
bounded by $|\F_\game|^2$, we know that the sequence $r_i$ is
stationary in at most $|\F_\game|^2+1$ steps, i.e., that there exists
$i\leq |\F_\game|^2+1$ such that $r_i=0$.

\begin{figure}\centering
  \begin{tikzpicture}
    \node[above] at (0,5) {$\Value_{\game}(\locMin,\valuation)$};
    \node[below] at (7,0) {$\valuation$};
    \draw[->] (-0.3,0) -- (7,0);
    \draw[->] (0,-0.3) -- (0,5);
    \draw[gray] (6,0) -- (6,5);
    \draw[gray] (3,0) -- (3,5);
    \draw[dashed] (7,4) -- (-0.5,0.25);
    \node[above] at (7,4) {$-\price(\locMin)$};
    
    \node[below] at (3,0) {$r_{i+1}$};
    \node[below] at (6,0) {$r_i$};
    
    \draw (6,3.5) -- (5,3);
    \draw[double] (5,3) --  (4,1.5);
    \draw (4,1.5) -- (3,1) -- (2,0.5) --  (1,-0.3);
    \node at (5,3) {$\bullet$};
    \node at (4,1.5) {$\bullet$};
    \node[left, text width = 2.5cm, text centered] at (0,2) {$\Value_\game(r_i) +$\\ $\price(\locMin) (r_{i}-r_{i+1})$};
    \draw[gray] (0,2) -- (3,2);
    
    \draw [dashed] (3,1) -- ++(-2.5,-1.25);

  \end{tikzpicture}
  \caption{The case $\locMin\in\LocsMin$: as the value at $r_{i+1}$ is strictly below $\Value_\game(r_i) + \price(\locMin) (r_{i}-r_{i+1})$, as the slope on the left of $r_i$ and of $r_{i+1}$ is $-\price(\locMin)$, there must exist a segment (represented with a double line) with slope greater than $-\price(\locMin)$ in $[r_{i+1},r_i)$. }
  \label{fig:slopesLocMin}
\end{figure}

\subsection{Proof of Lemma~\ref{lem:r_2-r_1-r_0}}
\label{app:r_2-r_1-r_0}

We denote by $r'$ the smallest valuation (smaller than
$r_1$) such that for all locations~$\loc$, $\Value_{\game}(\loc)$ is
affine over $[r',r_1]$. Then, the proof goes by contradiction: using
Lemma~\ref{lem:operatorNext}, we assume that for all
$\loc \in \Locs'\cap \LocsMin$ (respectively,
$\loc \in \Locs'\cap \LocsMax$):
\begin{itemize}
\item either ($\neg(i)$) the slope of $\Value_{\game}(\loc)$ on
  $[r',r_1]$ is greater (respectively, smaller) than
  $-\price(\loc)$,
\item or ($(i)\wedge \neg(ii)$) for all $\valuation\in [r',r_1]$,
  $\Value_{\game}(\loc,\valuation) = \Value_{\game}(\loc,r_0) +
  \price(\loc) (r_0-\valuation)$.
\end{itemize}

Let $\stratmin^0$ and $\stratmax^0$ (respectively, $\stratmin^1$ and
$\stratmax^1$) be a fake-optimal NC-strategy and an optimal
FP-strategy in $\game_{\Locs',r_0}$ (respectively,
$\game_{\Locs',r_1}$). Let
$r'' = \max (\points(\stratmin^1)\cup\points(\stratmax^1))\cap
[r',r_1)$,
so that strategies $\stratmin^1$ and $\stratmax^1$ have the
\emph{same behaviour} on all valuations of the interval $(r'',r_1)$,
i.e., either always play urgently the same transition, or wait, in
a non-urgent location, until reaching some valuation greater than or
equal to $r_1$ and then play the same transition.

Observe preliminarily that for all $\loc \in \Locs'\cap \LocsMin$
(respectively, $\loc \in \Locs'\cap \LocsMax$), if on the interval
$(r'',r_1)$, $\stratmin^1$ (respectively, $\stratmax^1$) goes to
$\loc^f$ then the slope on $[r'',r_1]$ (and thus on $[r',r_1]$) is
$-\price(\loc)$. Thus for such a location $\loc$, we know that
$(i)\wedge \neg(ii)$ holds for $\loc$ (by letting $r'$ be $r''$).

For other locations $\loc$, we will construct a new pair of NC- and
FP-strategies $\stratmin$ and $\stratmax$ in $\game_{\Locs',r_0}$
such that for all locations $\loc$ and valuations
$\valuation \in (r'',r_1)$
\begin{equation}\label{eq:fake-notfake}
  \fakeValue_{\game_{\Locs',r_0}}^{\stratmin}(\loc,\valuation) \leq
  \Value_{\game}(\loc,\valuation)\leq
  \Value_{\game_{\Locs',r_0}}^{\stratmax}(\loc,\valuation) \,.
\end{equation}
As a consequence, with Lemma~\ref{lem:fake-optimality} (over game
$\game_{\Locs',r_0}$), one would have that
$\Value_{\game_{\Locs',r_0}}(\loc,\valuation) =
\Value_\game(\loc,\valuation)$,
which will raise a contradiction with the definition of $r_1$ as
$\Next_{\Locs'}(r_0)<r_0$ (by Lemma~\ref{lem:operatorNext}), and
conclude the proof.

We only show the construction for $\stratmin$, as it is very similar
for $\stratmax$. Strategy $\stratmin$ is obtained by combining
strategies $\stratmin^1$ over $[0,r_1]$, and $\stratmin^0$ over
$[r_1,r_0]$: a special care has to be spent in case $\stratmin^1$
performs a jump to a location $\loc^f$, since then, in $\stratmin$,
we rather glue this move with the decision of strategy $\stratmin^0$
in $(\loc,r_1)$. Formally, let $(\loc,\valuation)$ be a
configuration of $\game_{\Locs',r_0}$ with $\loc \in \LocsMin$. We
construct $\stratmin(\loc,\valuation)$ as follows:
\begin{itemize}
\item if $\valuation \geq r_1$,
  $\stratmin(\loc,\valuation) = \stratmin^0(\loc,\valuation)$;
\item if $\valuation < r_1$, $\loc\not\in\Locs'$ and
  $\stratmin^1(\loc,\valuation) = \big(t,(\loc, \loc^f)\big)$ for
  some delay $t$ (such that $\valuation+t\leq r_1$), we let
  $\stratmin(\loc,\valuation) = \big (r_1-\valuation+t',(\loc,
  \loc') \big )$ where $(t',(\loc, \loc'))=\stratmin^0(\loc,r_1)$;
\item otherwise
  $\stratmin(\loc,\valuation) = \stratmin^1(\loc,\valuation)$.
\end{itemize}

For all finite plays $\run$ in $\game_{\Locs',r_0}$ that conform to
$\stratmin$, start in a configuration $(\loc,\valuation)$ such that
$\valuation\in (r'',r_0]$ and
$\loc \notin \{{\loc'}^f\mid \loc'\in \Locs\}$, and end in a final
location, we show by induction that
$\costgame{\game_{\Locs',r_0}}{\run} \leq
\Value_\game(\loc,\valuation)$.
Note that $\run$ either only contains valuations in $[r_1,r_0]$, or
is of the form
$(\loc,\valuation) \xrightarrow{c} (\loc^f,\valuation')$, or is of
the form $(\loc,\valuation)\xrightarrow{c} \run'$ with $\run'$ a run
that satisfies the above restriction.
\begin{itemize}
\item If $\valuation \in [r_1,r_0]$, then $\run$ conforms with
  $\stratmin^0$, thus, as $\stratmin^0$ is fake-optimal,
  $\costgame{\game_{\Locs',r_0}}{\run} \leq
  \Value_{\game_{\Locs',r_0}}(\loc,\valuation) =
  \Value_\game(\loc,\valuation)$
  (the last inequality comes from the definition of
  $r_1=\Next_{\Locs'}(r_0)$). Therefore, in the following cases, we
  assume that $\valuation \in (r'',r_1)$.
\item Consider then the case where $\run$ is of the form
  $(\loc,\valuation) \xrightarrow{c} (\loc^f,\valuation')$.
  \begin{itemize}
  \item if $\loc\in \Locs'\cap\LocsMin$, $\loc$ is urgent in
    $\game_{\Locs',r_0}$, thus
    $\valuation'=\valuation$. Furthermore, since $\run$ conforms
    with $\stratmin$, by construction of $\stratmin$, the choice of
    $\stratmin^1$ on $(r'',r_1)$ consists in going to $\loc^f$,
    thus, as observed above, $(i)\wedge \neg(ii)$ holds for
    $\loc$. Therefore,
    $\Value_{\game}(\loc,\valuation) = \Value_{\game}(\loc,r_0) +
    \price(\loc) (r_0-\valuation) = \fgoal_{\loc_f}(\valuation) =
    \costgame{\game_{\Locs',r_0}}{\run}$.
  \item If $\loc \in \LocsMin\setminus\Locs'$, by construction, it
    must be the case that
    $\stratmin(\loc,\valuation) = \big(r_1-\valuation+t',(\loc,
    \loc^f) \big)$
    where
    $\big(t,(\loc, \loc^f)\big) = \stratmin^1(\loc,\valuation)$ and
    $\big(t',(\loc, \loc^f)\big) = \stratmin^0(\loc,r_1)$.  Thus,
    $\valuation' = r_1 +t'$. In particular, observe that
    $\costgame{\game_{\Locs',r_0}}{\run} = (r_1-\valuation)
    \price(\loc) + \costgame{\game_{\Locs',r_0}}{\run'}$
    where
    $\run' = (\loc,r_1)\xrightarrow {c'} (\loc^f,\valuation')$. As
    $\run'$ conforms with $\stratmin^0$ which is fake-optimal in
    $\game_{\Locs',r_0}$,
    $\costgame{\game_{\Locs',r_0}}{\run'} \leq
    \Value_{\game_{\Locs',r_0}}(\loc,r_1) =
    \Value_{\game}(\loc,r_1)$
    (since $r_1=\Next(r_0)$). Thus
    $\costgame{\game_{\Locs',r_0}}{\run} \leq (r_1-\valuation)
    \price(\loc) + \Value_{\game}(\loc,r_1) =
    \costgame{\game_{\Locs',r_1}}{\run''}$
    where
    $\run''= (\loc,\valuation)
    \xrightarrow{c''}(\loc^f,\valuation+t)$
    conforms with $\stratmin^1$ which is fake-optimal in
    $\game_{\Locs',r_1}$. Therefore,
    $\costgame{\game_{\Locs',r_0}}{\run} \leq
    \Value_{\game_{\Locs',r_1}}(\loc,\valuation) =
    \Value_\game(\loc,\valuation)$ (since $r_1=\Next(r_0)$).
  \item If $\loc\in \LocsMax$ then
    $\costgame{\game_{\Locs',r_0}}{\run} =
    (\valuation'-\valuation)\price(\loc) +
    \fgoal_{\loc_f}(\valuation') =
    (\valuation'-\valuation)\price(\loc) + (r_0-\valuation')
    \price(\loc) + \Value_\game(\loc,r_0)= (r_0-\valuation)
    \price(\loc)+ \Value_\game(\loc,r_0)$.
    By Lemma~\ref{lem:rate}, since
    $\loc\in \LocsMax\setminus\LocsUrg$ ($\loc$ is not urgent in
    $\game$ since $\loc^f$ exists),
    $\Value_\game(\loc,r_1) \geq (r_0-r_1) \price(\loc)
    +\Value_\game(\loc,r_0)$.
    Furthermore, observe that if we define $\run'$ as the play
    $(\loc,\valuation)\xrightarrow{c'} (\loc^f,\valuation)$ in
    $\game_{\Locs',r_1}$, then $\run'$ conforms with $\stratmin^1$
    and
    \begin{align*}
      \costgame{\game_{\Locs',r_1}}{\run'} 
      & = (r_1-\valuation) \price(\loc) + \Value_\game(\loc,r_1) \\
      & \geq (r_1-\valuation) \price(\loc) + (r_0-r_1)
        \price(\loc) +\Value_\game(\loc,r_0) \\ 
      & = (r_0-\valuation) \price(\loc) + \Value_\game(\loc,r_0) \\
      & = \costgame{\game_{\Locs',r_0}}{\run}\,.
    \end{align*}
    Thus, as $\stratmin^1$ is fake-optimal in $\game_{\Locs',r_1}$,
    $\costgame{\game_{\Locs',r_0}}{\run} \leq
    \costgame{\game_{\Locs',r_1}}{\run'} \leq
    \Value_{\game_{\Locs',r_1}}(\loc,\valuation) =
    \Value_{\game}(\loc, \valuation)$.
  \end{itemize}
\item We finally consider the case where
  $\run = (\loc,\valuation)\xrightarrow{c}\run'$ with $\run'$ that
  starts in configuration $(\loc',\valuation')$ such that
  $\loc'\notin \{{\loc''}^f\mid \loc''\in \Locs\}$. By induction
  hypothesis
  $\costgame{\game_{\Locs',r_0}}{\run'}\leq
  \Value_\game(\loc',\valuation')$.
  \begin{itemize}
  \item If $\valuation' \leq r_1$, let $\run''$ be the play of
    $\game_{\Locs',r_1}$ starting in $(\loc',\valuation')$ that
    conforms with $\stratmin^1$ and $\stratmax^1$. If $\run''$ does
    not reach a final location, since $\stratmin^1$ is an
    NC-strategy, the costs of its prefixes tend to $-\infty$. By
    considering the strategy $\stratmin'$ of
    Lemma~\ref{lem:fake-optimality}, we would obtain a run
    conforming with $\stratmax^1$ of cost smaller than
    $\Value_{\game_{\Locs',r_1}}(\loc',\valuation')$ which would
    contradict the optimality of $\stratmax^1$. Hence, $\run''$
    reaches the target. Moreover, since $\stratmax^1$ is optimal and
    $\stratmin^1$ is fake-optimal, we finally know that
    $\costgame{\game_{\Locs',r_1}}{\run''} =
    \Value_{\game_{\Locs',r_1}}(\loc',\valuation')
    =\Value_{\game}(\loc',\valuation')$
    (since $\valuation'\in [\Next(r_1),r_1]$). Therefore,
    \begin{align*}
      \costgame{\game_{\Locs',r_0}}{\run} 
      & = (\valuation'-\valuation) \price(\loc) +
        \price(\loc,\loc') +\costgame{\game_{\Locs',r_0}}{\run'} \\ 
      & \leq (\valuation'-\valuation) \price(\loc) +
        \price(\loc,\loc') + 
        \Value_\game(\loc',\valuation') \\ 
      & = (\valuation'-\valuation) \price(\loc) + \price(\loc,\loc')
        + \cost{\run''} = \cost{(\loc,\valuation) \xrightarrow{c'} \run''} 
    \end{align*}
    Since the play $(\loc,\valuation) \xrightarrow{c'} \run''$
    conforms with $\stratmin^1$, we finally have
    $\costgame{\game_{\Locs',r_0}}{\run} \leq
    \cost{(\loc,\valuation) \xrightarrow{c'} \run''} \leq
    \Value_{\game_{\Locs',r_1}}(\loc,\valuation)=\Value_\game(\loc,\valuation)$.
  \item If $\valuation'> r_1$ and $\loc\in \LocsMax$, let $\run^1$
    be the play in $\game_{\Locs',r_1}$ defined by
    $\run^1 = (\loc,\valuation) \xrightarrow{c'}
    (\loc^f,\valuation)$
    and $\run^0$ the play in $\game_{\Locs',r_0}$ defined by
    $\run^0= (\loc,r_1) \xrightarrow{c''} \run'$. We have
    \begin{align*}
      \costgame{\game_{\Locs',r_0}}{\run} 
      &=
        (\valuation'-\valuation)\price(\loc) + \price(\loc,\loc') +
        \costgame{\game_{\Locs',r_0}}{\run'} \\
      &=
        \underbrace{\fgoal_{\loc_f}(\valuation)}_{
        =\costgame{\game_{\Locs',r_1}}{\run^1}}-\Value_\game(\loc,r_1) 
        + 
        \underbrace{(\valuation'-r_1)\price(\loc) + \price(\loc,\loc') +
        \costgame{\game_{\Locs',r_0}}{\run'}}_{
        =\costgame{\game_{\Locs',r_0}}{\run^0}}\,. 
    \end{align*}
    Since $\run^{0}$ conforms with $\stratmin^0$, fake-optimal, and
    reaches a final location,
    $\costgame{\game_{\Locs',r_0}}{\run^0}\leq
    \Value_{\game_{\Locs',r_0}}(\loc,r_1) = \Value_\game(\loc,r_1)$
    (since $r_1=\Next_{\Locs'}(r_0)$).  We also have that $\run^1$
    conforms with $\stratmin^1$, so the previous explanations
    already proved that
    $\costgame{\game_{\Locs',r_1}}{\run^1} \leq
    \Value_\game(\loc,\valuation)$.
    As a consequence
    $\costgame{\game_{\Locs',r_0}}{\run} \leq
    \Value_\game(\loc,\valuation)$.
  \item If $\valuation'> r_1$ and $\loc\in \LocsMin$, we know that
    $\loc$ is non-urgent, so that $\loc\not\in \Locs'$. Therefore,
    by definition of $\stratmin$,
    $\stratmin(\loc,\valuation) = ( r_1-\valuation+t',(\loc,\loc'))$
    where $\stratmin^1(\loc,\valuation) = (t,(\loc, \loc^f))$ for
    some delay $t$, and
    $\stratmin^0(\loc,r_1) = (t',(\loc, \loc'))$. If we let $\run^1$
    be the play in $\game_{\Locs',r_1}$ defined by
    $\run^1 = (\loc,\valuation) \xrightarrow{c'}(\loc^f,\valuation)$
    and $\run^0$ the play in $\game_{\Locs',r_0}$ defined by
    $\run^0= (\loc,r_1)\xrightarrow{c''} \run'$, as in the previous
    case, we obtain that
    $\costgame{\game_{\Locs',r_0}}{\run} \leq
    \Value_\game(\loc,\valuation)$.
  \end{itemize}
\end{itemize}

As a consequence of this induction, we have shown that for all
$\loc\in \Locs$, and for all $\valuation\in (r'',r_1)$,
$\fakeValue_{\game_{\Locs',r_0}}^{\stratmin}(\loc,\valuation) \leq
\Value_\game(\loc,\valuation)$,
which shows one inequality of \eqref{eq:fake-notfake}, the other
being obtained very similarly.

\section{Run of the algorithm on an example}
\label{app:example-running}

\begin{figure}[tbp]
  \begin{center}
    \begin{tikzpicture}[xscale=.8,yscale=0.55]

      \draw[->] (6,0) -- (11,0) node[anchor=north] {$x$};
      \draw	(6,0) node[anchor=south] {$0$}
      (7,0) node[anchor=south] {$\frac 1 4$}
      (8,0) node[anchor=south] {$\frac 1 2$}
      (9,0) node[anchor=south] {$\frac 3 4$}
      (10,0) node[anchor=north] {$1$};

      \draw[->] (6,0) -- (6,-4) node[anchor=east] {$\val(\loc_2,x)$};
      \draw	(6,-3.3) node[anchor=east] {$-9.5$}
      (6,-2) node[anchor=east] {$-6$}
      (6.05,-1.85) node[anchor=west] {$-5.5$}
      (6,-0.6) node[anchor=east] {$-2$}
      (6,0.3) node[anchor=east] {$1$};

      \draw[thick] (6,-3.3) -- (7,-2) -- (8,-1.85)--(9,-0.6)--(10,0.3);

      \draw[->] (6,-5) -- (11,-5) node[anchor=north] {$x$};
      \draw	(6,-5) node[anchor=south] {$0$}
      (7,-5) node[anchor=south] {$\frac 1 4$}
      (8,-5) node[anchor=south] {$\frac 1 2$}
      (9,-5) node[anchor=south] {$\frac 3 4$}
      (9.6,-5) node[anchor=south] {$\frac 9 {10}$}
      (10,-5) node[anchor=south] {$1$};

      \draw[->] (6,-5) -- (6,-9) node[anchor=east] {$\val(\loc_1,x)$};
      \draw	(6,-8.3) node[anchor=east] {$-9.5$}
      (6,-7) node[anchor=east] {$-6$}
      (6.05,-6.85) node[anchor=west] {$-5.5$}
      (6,-5.6) node[anchor=east] {$-2$}
      (6,-5.1) node[anchor=east] {$-0.2$};

      \draw[thick] (6,-8.3) -- (7,-7) -- (8,-6.85)--(9,-5.6)--(9.6,-5.1)--(10,-5);

      \draw[->] (-2,0) -- (3,0) node[anchor=north] {$x$};
      \draw	(-2,0) node[anchor=south] {$0$}
      (-1,0) node[anchor=south] {$\frac 1 4$}
      (0,0) node[anchor=south] {$\frac 1 2$}
      (2,0) node[anchor=south] {$1$};

      \draw[->] (-2,0) -- (-2,-4) node[anchor=east] {$\val(\loc_3,x)$};
      \draw	(-2,-3) node[anchor=east] {$-10$}
      (-2,-1.7) node[anchor=east] {$-6$}
      (-2.05,-1.5) node[anchor=west] {$-5.5$}
      (-2,-2) node[anchor=west] {$-7$};

      \draw[thick] (-2,-3) -- (-1,-1.7) -- (0,-1.5)--(2,-2);

      \draw[->] (-2,-5) -- (3,-5) node[anchor=north] {$x$};
      \draw	(-2,-5) node[anchor=south] {$0$}
      (2,-5) node[anchor=south] {$1$};

      \draw[->] (-2,-5) -- (-2,-9) node[anchor=east] {$\val(\loc_4,x)$};
      \draw	(-2,-6.3) node[anchor=east] {$-4$}
      (-2,-7.3) node[anchor=east] {$-7$};

      \draw[thick] (-2,-6.3) -- (2,-7.3);

      \draw[->] (6,-10) -- (11,-10) node[anchor=north] {$x$};
      \draw	(6,-10) node[anchor=south] {$0$}
      (9,-10) node[anchor=south] {$\frac 3 4$}
      (10,-10) node[anchor=south] {$1$};

      \draw[->] (6,-10) -- (6,-14) node[anchor=east] {$\val(\loc_5,x)$};
      \draw	(6,-13.5) node[anchor=east] {$-14$}
      (6,-10.5) node[anchor=east] {$-2$}
      (6,-9.8) node[anchor=east] {$1$};

      \draw[thick] (6,-13.5) -- (9,-10.5) --(10,-9.8);

      \draw[->] (6,-15) -- (11,-15) node[anchor=north] {$x$};
      \draw	(6,-15) node[anchor=south] {$0$}
      (10,-15) node[anchor=south] {$1$};

      \draw[->] (6,-15) -- (6,-19) node[anchor=east] {$\val(\loc_6,x)$};
      \draw	(6,-18) node[anchor=east] {$-11$}
      (6,-14.8) node[anchor=east] {$1$};

      \draw[thick] (6,-18) -- (10,-14.8);

      \draw[->] (-2,-10) -- (3,-10) node[anchor=north] {$x$};
      \draw	(-2,-10) node[anchor=south] {$0$}
      (2,-10) node[anchor=south] {$1$};

      \draw[->] (-2,-10) -- (-2,-14) node[anchor=east] {$\val(\loc_7,x)$};
      \draw	(-2,-13.5) node[anchor=east] {$-16$}
      (-2,-10) node[anchor=east] {$0$};

      \draw[thick] (-2,-13.5) -- (2,-10);

    \end{tikzpicture}

    \caption{Value functions of the \SPTG of Figure~\ref{fig:ex-ptg2}}
    \label{fig:val_sptg}
  \end{center}
\end{figure}

Figure~\ref{fig:val_sptg} shows the value functions of the \SPTG of
Figure~\ref{fig:ex-ptg2}. Here is how the algorithm obtains those
functions.  First it computes the functions at valuation $1$, thanks
to \SolveInstant. Then, it computes the value of the game where all
states are urgent but additional terminal states have been added by
the \Waiting function to allow waiting until $1$. This step gives the
correct value functions until the cutpoint $\frac 3 4$: in the
$repeat$ loop, at first $a= 9/10$ but the slope in $\loc_1$ is smaller
than the slope that would be granted by waiting. Then $a=3/4$ where
the algorithm gives a slope of value $-16$ in $\loc_2$ while the cost
of this $\MaxPl$'s location is $-14$. We thus choose $r:=3/4$ and
compute the algorithm on the interval $[0, r]$ with final states
allowing one to wait until $r$ and get the already known value in
$r$. The algorithm then stops at $\frac 1 2$ in order to allow
$\loc_1$ to wait, then in $\frac 1 4$ because of $\loc_2$ and finally
the algorithm reaches $0$ giving us the value functions on the entire
interval $[0,1]$.

\section{Reset-acyclic \PTG{s}}
\label{app:raptg}

Towards solving reset-acyclic \PTG{s}, our first step is to remove
strict guards from the transitions, i.e., guards of the form $(a,b]$,
$[b,a)$ or $(a,b)$ with $a,b\in \N$. For this, we enhance the \PTG
with regions in a method similar to what is done
in~\cite[Lemma~4.6]{DueIbs13}. Formally, let
$\game=
(\LocsMin,\LocsMax,\LocsFin,\LocsUrg,\fgoalvec,\transitions,\price)$
be a \PTG. We define the region-\PTG of $\game$ as
$\game'=
(\LocsMin',\LocsMax',\LocsFin',\LocsUrg',\fgoalvec',\transitions',\price')$
where:
\begin{itemize}
\item $\LocsMin'=\{(\loc,I)\mid \loc\in \LocsMin, I\in \reggame\}$;
\item $\LocsMax'=\{(\loc,I)\mid \loc\in \LocsMax, I\in \reggame\}$;
\item $\LocsFin=\{(\loc,I)\mid \loc\in \LocsFin, I\in \reggame\}$;
\item $\LocsUrg=\{(\loc,I)\mid \loc\in \LocsUrg, I\in \reggame\}$;
\item $\forall (\loc,I)\in \LocsFin', \fgoal'_{\loc,I}=\fgoal_\loc$;
\item \todo{B: isn't it the case that $I'_g = \overline{I_g\cap I}$ in
    every cases?} \begin{align*}
    \transitions'
    &=\Bigg\{((\loc,I),\overline{I_g\cap I},R,(\loc',I'))\mid
      (\loc,I_g,R,\loc')\in \transitions, I' =
      {\footnotesize\begin{cases}
          I &\text{if } R=\bot\\
          \{0\} &\text{otherwise}
      \end{cases}}
                                         \Bigg\} \\
    & \quad\cup\big\{((\loc,(M_k,M_{k+1})),\{M_{k+1}\},\bot,(\loc,\{M_{k+1}\}))
      \mid \loc\in \Locs, (M_k,M_{k+1}) \in \reggame\big\} \\
    & \quad\cup\big\{((\loc,\{M_{k}\}),\{M_{k}\},\bot,(\loc,(M_k,M_{k+1})))
      \mid \loc\in \Locs, (M_k,M_{k+1}) \in \reggame\big\}\,;
  \end{align*}
\item $\forall (\loc,I)\in \Locs', \price'(\loc,I) = \price(\loc)$;
  and $\forall ((\loc,I),I_g,R,(\loc',I'))\in \transitions'$, if
  $(\loc,I_g,R,\loc')\in \transitions$, then
  $\price((\loc,I),I_g,R,(\loc',I'))=\price(\loc,I_g,R,\loc)$,
  else $\price((\loc,I),I_g,R,(\loc',I'))=0$.
\end{itemize}

It is easy to verify that, in all configurations
$((\loc,\{M_k\}),\valuation)$ reachable from the null valuation, the
valuation $\valuation$ is $M_k$. More interestingly, in all
configurations $((\loc,(M_k,M_{k+1})),\valuation)$ reachable from the
null valuation, the valuation $\valuation$ is in $[M_k, M_{k+1}]$:
indeed if $\valuation=M_k$ (respectively, $M_{k+1}$), it intuitively
simulates a configuration of the original game with a valuation
arbitrarily close to $M_k$, but greater than $M_k$ (respectively,
smaller than $M_{k+1}$). The game can thus take transitions with guard
$x>M_k$, but can not take transitions with guard $x=M_k$ anymore.

\begin{lemma}
  \label{lem:region_ptg}
  Let $\game$ be a one-clock \PTG, and $\game'$ be its region-\PTG
  defined as before. For $(\loc,I)\in \Locs\times\reggame$ and
  $\valuation\in I$,
  $\val_\game(\loc,\valuation)=\val_{\game'}((\loc,I),\valuation)$.
  Moreover, we can transform an $\varepsilon$-optimal strategy of
  $\game'$ into a $\varepsilon'$-optimal strategy of $\game$ with
  $\varepsilon' > \varepsilon$.
\end{lemma}
\begin{proof}
  The proof consists in replacing strategies of $\game'$ where players
  can play on the borders of regions, by strategies of $\game$ that
  play increasingly close to the border as time passes. If played
  close enough, the loss created can be chosen as small as we want. 
\end{proof}

Consider now the region-\PTG $\game$ associated to a reset-acyclic
\PTG (and of polynomial size with respect to the original \PTG). We
can decompose the graph of $\game$ into strongly connected components
(that do not contain reset transitions by hypothesis). Consider first
its bottom strongly connected components, i.e., components with no
reset transitions exiting from them. All clock constraints are of the
form $[a,b]$ with $a<b$, or $\{a\}$. We denote by
$0=M_0<M_1<\cdots< M_K$ the constants appearing in the guards of the
component (adding $0$). Then, solving the component amounts
to \begin{inparaenum}[$(i)$]
\item solve the sub-game with only transitions with guard $\{M_k\}$,
  replacing then these transitions by final locations with the cost
  just computed,
\item solve the modified sub-game with only transitions with guard
  $[M_{k-1},M_k]$, by first shrinking the guards to transform the game
  into an \SPTG, and so on, until $M_0=0$. 
\end{inparaenum}
Once all bottom strongly connected components are solved, we replace
the reset transitions going to them by final locations again, using
the cost computed so far. We continue until no strongly connected
components remain.  Each \SPTG being solvable in exponential time, the
overall reset-acyclic can be solved in exponential time too.

\end{document}